\documentclass[12pt]{article}
\usepackage{amssymb,amsfonts,amsmath,amsthm}
\usepackage{braket}
\usepackage{booktabs}	
\usepackage{float}		
\usepackage{graphicx}		
\usepackage{caption}		
\usepackage{subcaption}	

\title{Optimal Reinsurance and Investment in a Diffusion Model}
\author{\begin{minipage}{0.45\linewidth}
{\sc Matteo Brachetta$^1$}\\ \small Department of Economics\\ University of Pescara\\ Viale Pindaro, 42 \\ 65127 Pescara,
Italy 
\end{minipage}
\and
\begin{minipage}{0.45\linewidth}{\sc Hanspeter Schmidli$^2$}\\ \small Institute of
Mathematics\\ University of Cologne\\ Weyertal 86--90\\ 50931 K\"oln,
Germany\end{minipage}}

\makeatletter
\def\mathbox#1{\mbox{\let\@@beforemath\relax \m@th $#1$}}
\def\id{\;{\rm d}}
\def\dd{{\rm d}}
\def\E{{\rm I \mkern-2.5mu \nonscript\mkern-.5mu I \mkern-6.5mu E}}
\def\Prob{{\rm I \mkern-2.5mu \nonscript\mkern-.5mu I \mkern-6.5mu P}}
\def\R{{\rm I \mkern-2.5mu \nonscript\mkern-.5mu R}}

\def\e{{\rm e}}
\def\onehalf{\protect{\mathbox{\frac12}}}
\let\epsilon\varepsilon
\let\phi\varphi
\newcommand{\takis}{\usefont{OMS}{ztmcm}{m}{n}}
\newcommand{\nicecal}[1]{\text{\takis #1}}
\def\caldef#1{\expandafter\def\csname cal#1\endcsname{{\nicecal{#1}}}}
\caldef F
\def\pderiv#1#2{\frac{\partial #1}{\partial #2}}
\newif\ifpdf
\@ifundefined{pdfoutput}{\pdffalse}{\ifnum\pdfoutput<1 \pdffalse\else\pdftrue
\fi}
\makeatother

\newtheorem{theo}{Theorem}
\newtheorem{lemma}{Lemma}
\theoremstyle{plain}
\newtheorem{remark}{Remark}
\theoremstyle{plain}
\newtheorem{prop}{Proposition}
\theoremstyle{plain}
\newtheorem{corollary}{Corollary}

\begin{document}
\maketitle
\thispagestyle{empty}
\stepcounter{footnote}\footnotetext{e-mail: \tt matteo.brachetta@unich.it}

\stepcounter{footnote}\footnotetext{e-mail: \tt schmidli@math.uni-koeln.de}

\begin{abstract}\noindent
We consider a diffusion approximation to an insurance risk model where an
external driver models a stochastic environment. The insurer can buy reinsurance.
Moreover, investment in a financial market is possible. The financial market
is also driven by the environmental process. Our goal is to maximise terminal
expected utility. In particular, we consider the case of SAHARA utility
functions. In the case of proportional and excess-of-loss reinsurance, we
obtain explicit results.

\medskip


\medskip

\noindent {Classification:} Primary 91B30; secondary 60G44; 60J60; 93E20.

\end{abstract}
\newpage


\section{Introduction}

The optimal reinsurance-investment problem is of large interest in the
actuarial literature. A reinsurance is a contract whereby a reinsurance
company agrees to indemnify the cedent (i.e. the primary insurer) against all
or part of the future losses that the latter sustains under the policies that
she has issued. For this service the insurer is asked to pay a premium. It is
well known that such a risk-sharing agreement allows the insurer to reduce the
risk, to increase the business capacity, to stabilise the operating results and so on.
In the existing literature there are a lot of works dealing with the optimal
reinsurance strategy, starting from the seminal papers \cite{definetti:1940},
\cite{buhlmann:1970} and \cite{gerber:1979}. During the last decades two
different approaches were used to study the problem: some authors model the
insurer's surplus as a jump process, others as a diffusion approximation (see
e.g. \cite{schmidli:2018risk} and references therein for details about
risk models). In addition, only two reinsurance agreements were considered:
the proportional and the excess-of-loss contracts (or both, as a mixed
contract). Among the optimization criteria, we recall the expected utility
maximization (see \cite{irgens_paulsen:optcontrol}, \cite{GUERRA2008529} and
\cite{Mania2010}), ruin probability minimization (see \cite{promislow2005},
\cite{schmidli:2001} and
\cite{schmidli2002}), dividend policy optimization
(see \cite{buhlmann:1970} and \cite{schmidli:control}) and others. In
particular, the former was developed only for CRRA and CARA utility
functions.

Our aim is to investigate the optimal reinsurance problem in a diffusion risk model when the insurer subscribes a general reinsurance agreement, with a retention level $u\in[0,I]$. The insurer's objective is to maximise the expected utility of the terminal wealth for a general utility function $U$, satisfying the classical assumptions (monotonicity and concavity). That is, we do not assume any explicit expression neither for the reinsurance policy nor for $U$. However, we also investigate how our general results apply to specific utility functions, including CRRA and CARA classes, and to the most popular reinsurance agreements such as proportional and excess-of-loss.

One additional feature of our paper is that the insurer's surplus is affected
by an environmental factor $Y$, which allows our framework to take into
account \textit{size} and \textit{risk fluctuations} (see \cite[Chapter
2]{grandell:risk}). We recall two main attempts of introducing a stochastic
factor in the risk model dynamic: in \cite{liangbayraktar:optreins} the
authors considered a Markov chain with a finite state space, while in
\cite{BC:IME2019}
$Y$ is a diffusion process, as in our case. However, they
considered jump processes and the rest of the model formulation is very
different (for instance, they restricted the maximization to the exponential
utility function and the proportional reinsurance). Moreover, in those papers
$Y$ only affects the insurance market.

Indeed, another important peculiarity of our model is the dependence between the insurance and the financial markets. We allow the insurer to invest her money in a risky asset, modelled as a diffusion process with both the drift and the volatility influenced by the stochastic factor $Y$. From the practical point of view, this characteristic reflects any connection between the two markets. From the theoretical point of view, we remove the standard assumption of the independence, which is constantly present in all the previous works, especially because it simplifies the mathematical framework.

The paper is organized as follows: in the following section we formulate our optimal stochastic control problem; next, in Section \ref{section:valuefun} we analyse the main properties of the value function, while in Section \ref{section:hjb} we characterize the value function as a viscosity solution to the Hamilton-Jacobi-Bellman (HJB) equation associated with our problem; in Section \ref{section:sahara} we apply our general results to the class of SAHARA utility functions, which includes CRRA and CARA utility functions as limiting cases. In addition, we characterize the optimal reinsurance strategy under the proportional and the excess-of-loss contracts, also providing explicit formulae. Finally, in Section \ref{section:numerical} we give some numerical examples.


\section{The Model}

The surplus process of an insurer is modelled as the solution to the
stochastic differential equation
\[ \dd X_t^0 = m(t, Y_t, u_t) \id t + \sigma(t,Y_t, u_t) \id W_t^1\;,\qquad
X_0^0 = x\;,\]
where
$Y$ is an environmental process, satisfying
\[ \dd Y_t = \mu_Y(t,Y_t) \id t + \sigma_Y(t,Y_t) \id W_t^Y\;,\qquad Y_0 =
y\;,\]
and $u_t$ is the reinsurance retention level of the insurer at time
$t$. We assume that $u_t$ is a cadlag process and can take all the values in
an interval $[0,I]$,
where $I \in (0,\infty]$ and that the functions $m(t,y,u)$, $\sigma(t,y,u)$,
$\mu_Y(t,y)$ and $\sigma_Y(t,y)$
are continuously differentiable bounded functions satisfying a Lipschitz
condition uniformly in $u$. 
Further, the insurer has the possibility to invest into a risky asset $R$
modelled as the solution to
\[ \dd R_t = \mu(t,Y_t) R_t \id t + \sigma_1(t,Y_t) R_t \id W_t^1 +
\sigma_2(t,Y_t) R_t \id W_t^2\;,\quad R_0 \in(0,+\infty)\;.\]
Also the functions $\mu(t,y)$, $\sigma_1(t,y)$ and $\sigma_2(t,y)$ are assumed
to be bounded continuous positive functions satisfying a Lipschitz condition. We
further assume that $\sigma_1(t,y) + \sigma_2(t,y)$ is bounded away from
zero. Here, $W^1$, $W^2$, $W^Y$ are independent Brownian motions on a 
reference probability space $(\Omega, \calF, \Prob)$.
Thus, the reinsurance strategy does not influence the behaviour of the risky
asset. But, the surplus process and the risky asset are dependent. Choosing an
investment strategy $a$, the surplus of the insurer fulfils
\begin{eqnarray*}
\dd X_t^{u,a} &=& \{m(t,Y_t,u_t) + a_t \mu(t,Y_t)\} \id t + \{(\sigma(t,Y_t, u_t)
+a_t \sigma_1(t,Y_t))\} \id W_t^1\\ 
&& \hskip5.4cm {}+ a_t \sigma_2(t,Y_t) \id W_t^2\;, \qquad X_0^{u,a} = x\;.
\end{eqnarray*}
In order that a strong solution exists we assume that $\E[\int_0^T a_t^2 \id
t] < \infty$. 
Our goal is to maximise the terminal expected utility at time $T > 0$
\[ V^{u,a}(0,x,y) = \E[U(X_T^{u,a})\mid X_0^{u,a} = x, Y_0 = y]\;,\]
and, if it exists, to find the optimal strategy $(u^*, a^*)$. That is,
\[ V(0,x,y) = \sup_{u,a} V^{u,a}(0,x,y) = V^{u^*,a^*}(0,x,y)\;,\]
where the supremum is taken over all measurable adapted processes $(u,a)$ such
that the conditions above are fulfilled. $U$ is a
utility function. That is, $U$ is strictly increasing and strictly concave. We
make the additional assumption that $U''(x)$ is continuous. The
filtration is the smallest complete right-continuous filtration $\{\calF_t\}$
such that the Brownian motions are adapted. In particular, we suppose that $Y$
is observable.

We will also need the value functions if we do start at time $t$ instead. Thus
we define
\[ V^{u,a}(t,x,y) = \E[U(X_T^{u,a})\mid X_t^{u,a} = x, Y_t = y]\;,\]
where we only consider strategies on the time interval $[t,T]$
and, analogously, $V(t,x,y) = \sup_{u,a} V^{u,a}(t,x,y)$. The boundary
condition is then $V(T,x,y) = U(x)$. Because our underlying processes are
Markovian, $V(t,X_t^{u,a},Y_t)$ depends on $\calF_t$ via $(X_t^{u,a}, Y_t)$ only.


\section{Properties of the value function}
\label{section:valuefun}

\begin{lemma}
\begin{enumerate}
\item The value function is increasing in $x$.
\item The value function is continuous.
\end{enumerate}

\end{lemma}
\begin{proof}
That the value function is increasing in $x$ is clear. By It\^o's formula
\begin{eqnarray*}
U(X_T^{u,a}) &=& U(x) + \int_t^T [\{ m(s,Y_s,u_s) + a_s \mu(s,Y_s)\}
U'(X_s^{u,a})\\ 
&&\hskip5mm {}+ \onehalf\{(\sigma(s,Y_s,u_s) + a_s \sigma_1(s,Y_s))^2 +
a_s^2 \sigma_2^2(s,Y_s) \}U''(X_s^{u,a})] \id s\\ 
&&{}+ \int_t^T(\sigma(s,Y_s,u_s) +
a_s \sigma_1(s,Y_s))  \id W_s^1 + \int_t^T a_s \sigma_2(s,Y_s) \id W_s^2\;.
\end{eqnarray*}
Because the stochastic integrals are martingales by our assumptions
\begin{eqnarray*}
\E[U(X_T^{u,a})] &=& U(x) + \E\Bigl[\int_t^T [\{ m(s,Y_s,u_s) + a_s \mu(s,Y_s)\}
U'(X_s^{u,a})\\ 
&&\hskip5mm{} + \onehalf\{(\sigma(s,Y_s,u_s) + a_s \sigma_1(s,Y_s))^2\\ 
&&\hskip15mm{} + a_s^2 \sigma_2^2(s,Y_s) \}U''(X_s^{u,a})] \id s\Bigr]\;.
\end{eqnarray*}
Taking the supremum over the strategies we get the continuity by the Lipschitz
assumptions.
\end{proof}

\begin{lemma}
The value function is concave in $x$.
\end{lemma}
\begin{proof}
If the value function was not concave, we would find $x$ and a test
function $\phi$ with $\phi_{x x}(t,x,y) \ge 0$, $\phi_x(t,x,y) > 0$ and
$\phi(t',x',y') \le V(t',x',y')$ for all $t'$, $x'$, $y'$ and $\phi(t,x,y) =
V(t,x,y)$. By the proof of Theorem \ref{thm:vis.sol} below,
\begin{eqnarray*}
0 &\ge& \phi_t + \sup_{u,a} \{m(t,y,u) + a \mu(t,y)\} \phi_x\\ 
&&{}+ \onehalf \{(\sigma(t,y, u) + a \sigma_1(t,y))^2 +  a^2 \sigma_2^2(t,y)\}
\phi_{x x} + \mu_Y(t,y) \phi_y\\ 
&&{}+ \onehalf \sigma_Y^2(t,y) \phi_{y y}\;.
\end{eqnarray*}
But it is possible to choose $a$ such that the above inequality does not hold.
\end{proof}


\section{The HJB equation}
\label{section:hjb}

We expect the value function to solve
\begin{eqnarray}\label{eq:hjb}
0 &=& V_t + \sup_{u,a} \{m(t,y,u) + a \mu(t,y)\} V_x\nonumber\\ 
&&{}+ \onehalf \{(\sigma(t,y, u) + a \sigma_1(t,y))^2 +  a^2 \sigma_2^2(t,y)\}
V_{x x} + \mu_Y(t,y) V_y\nonumber\\ 
&&{}+ \onehalf \sigma_Y^2(t,y) V_{y y}\;.
\end{eqnarray}
A (classical) solution is only possible if $V_{x x} < 0$. In this case,
\begin{equation}
\label{eqn:a_optimal}
a = - \frac{\mu(t,y) V_x + \sigma(t,y,u) \sigma_1(t,y) V_{x
x}}{(\sigma_1^2(t,y)+ \sigma_2^2(t,y)) V_{x x}}\;.
\end{equation}
Thus, we need to solve
\begin{eqnarray}\label{eqn:hjb2}
0 &=& V_t + \sup_{u} m(t,y,u) V_x - \frac{(\mu(t,y) V_x + \sigma(t,y,u)
\sigma_1(t,y) V_{x x})^2}{2(\sigma_1^2(t,y)+ \sigma_2^2(t,y)) V_{x x}}\nonumber\\ 
&&{}+ \frac{1}{2}\sigma^2(t,y,u) V_{x x}+ \mu_Y(t,y) V_y
+ \onehalf \sigma_Y^2(t,y) V_{y y}\;.
\end{eqnarray}
By our assumption that $m(t,y,u)$ and $\sigma(t,y,u)$ are continuous functions
on a closed interval of the compact set $[0,\infty]$, there is a value
$u(x,y)$ at which that supremum is 
taken.

\begin{theo}\label{thm:vis.sol}
The value function is a viscosity solution to \eqref{eq:hjb}.
\end{theo}
\begin{proof}
Without loss of generality we only show the assertion for $t = 0$.
Choose $(\bar u,\bar a)$ and $\epsilon, \delta, h > 0$. Let $\tau^{\bar u,
\bar a} = \inf\{ t > 0: \max\{|X_t^{\bar u, \bar a} - x|, |Y_t-y|\} >
\epsilon\}$ and $\tau = \tau^{\bar u, \bar a} \wedge h$. Consider the 
following strategy. $(u_t,a_t) = (\bar u, \bar a)$ for $t < \tau^{\bar u,
\bar a} \wedge h$, and $(u_t, a_t) = (\tilde u_{t-(\tau^{\bar u,
\bar a} \wedge h)}, \tilde u_{t-(\tau)})$ for some
strategy $(\tilde u, \tilde a)$, such that $V^{\tilde u, \tilde
a}(\tau, X_{\tau}^{\bar u,
\bar a},Y_{\tau^{\bar u, \bar a} \wedge
h}) > V(\tau, X_{\tau^{\bar u,
\bar a} \wedge h}^{\bar u, \bar a},Y_{\tau^{\bar u, \bar a} \wedge
h})-\delta$. Note that the strategy can be chosen in a measurable way since
$V(t,x,y)$ is continuous. Let $\phi(t,x,y)$ be a test function, such that
$\phi(t,x',y') \le V(t,x',y')$ with $\phi(0,x,y) = V(0,x,y)$. Then by It\^o's
formula
\begin{eqnarray*}
\phi(\tau, X_{\tau}^{\bar u,
\bar a},Y_{\tau}) &=& \phi(0,x,y) + \int_0^\tau [\phi_t(t,X_t, Y_t)\\ 
&&{}+ \{m(t, Y_t, \bar u) + \bar a \mu(t,Y_t)\} \phi_x(t,X_t, Y_t)\\ 
&&{}+ \onehalf\{ (\sigma(t,Y_t,\bar u) + \bar a \sigma_1(t,Y_t))^2 + \bar a^2
\sigma_2^2(t, Y_t)\} \phi_{x x}(t,X_t, Y_t)\\ 
&&{}+  \mu_Y(t,Y_t) \phi_y(t,X_t,Y_t) + \onehalf \sigma_Y^2(t,Y_t) \phi_{y
y}(t,X_t,Y_t)] \id t\\ 
&&{}+ \int_0^\tau [\sigma(t,Y_t,\bar u) + \bar a
\sigma_1(t,Y_t)] \phi_x(t,X_t, Y_t) \id W_t^1\\ 
&&{}+ \int_0^\tau \bar a \sigma_2(t,Y_t) \phi_x(t,X_t, Y_t)\id W_t^2\\
&&{}+ \int_0^\tau \sigma_Y(t,Y_t) \phi_y(t,X_t,Y_t) \id W_t^Y\;.
\end{eqnarray*}
Note that the integrals with respect to the Brownian motions are true
martingales since the derivatives of $\phi$ are continuous and thus bounded on
the (closed) area, and therefore the integrands are bounded. Taking expected
values gives
\begin{eqnarray*}
V(0,x,y) &\ge& V^{u,a}(0,x,y) = \E[V^{u,a}(\tau, X_\tau, Y_\tau)] \ge \E[V(\tau,
X_\tau, Y_\tau)] - \delta\\ 
&\ge& \E[\phi(\tau,X_\tau, Y_\tau)] - \delta\\
&=& V(0,x,y) -\delta + \E\Bigl[ \int_0^\tau [\phi_t(t,X_t, Y_t)\\ 
&&{}+ \{m(t, Y_t, \bar u) + \bar a \mu(t,Y_t)\} \phi_x(t,X_t, Y_t)\\ 
&&{}+ \onehalf\{ (\sigma(t,Y_t,\bar u) + \bar a \sigma_1(t,Y_t))^2 + \bar a^2
\sigma_2^2(t, Y_t)\} \phi_{x x}(t,X_t, Y_t)\\ 
&&{}+  \mu_Y(t,Y_t) \phi_y(t,X_t,Y_t) + \onehalf \sigma_Y^2(t,Y_t) \phi_{y
y}(t,X_t,Y_t)] \id t\Bigr]\;.
\end{eqnarray*}
The right hand side does not depend on $\delta$. We thus can let $\delta =
0$. This yields
\begin{eqnarray*}
0 &\ge& \E\Bigl[ \frac1h\int_0^\tau [\phi_t(t,X_t, Y_t)
+ \{m(t, Y_t, \bar u) + \bar a \mu(t,Y_t)\} \phi_x(t,X_t, Y_t)\\ 
&&{}+ \onehalf\{ (\sigma(t,Y_t,\bar u) + \bar a \sigma_1(t,Y_t))^2 + \bar a^2
\sigma_2^2(t, Y_t)\} \phi_{x x}(t,X_t, Y_t)\\ 
&&{}+  \mu_Y(t,Y_t) \phi_y(t,X_t,Y_t) + \onehalf \sigma_Y^2(t,Y_t) \phi_{y
y}(t,X_t,Y_t)] \id t\Bigr]\;.
\end{eqnarray*}
It is well known that $h \Prob[\tau \le h]$ tends to zero as $h
\downarrow 0$. Thus, letting $h \downarrow 0$ gives
\begin{eqnarray*}
0 &\ge& \phi_t + \{m(t,y,\bar u) + \bar a \mu(t,y)\} \phi_x\nonumber\\ 
&&{}+ \onehalf \{(\sigma(t,y, \bar u) + \bar a \sigma_1(t,y))^2 +  a^2
\sigma_2^2(t,y)\} \phi_{x x} + \mu_Y(t,y) \phi_y\nonumber\\ 
&&{}+ \onehalf \sigma_Y^2(t,y) \phi_{y y}\;.
\end{eqnarray*}
Since $(\bar u, \bar a)$ is arbitrary,
\begin{eqnarray*}
0 &\ge& \phi_t + \sup_{u,a} \{m(t,y,u) + a \mu(t,y)\} \phi_x\nonumber\\ 
&&{}+ \onehalf \{(\sigma(t,y, u) + a \sigma_1(t,y))^2 +  a^2 \sigma_2^2(t,y)\}
\phi_{x x} + \mu_Y(t,y) \phi_y\nonumber\\ 
&&{}+ \onehalf \sigma_Y^2(t,y) \phi_{y y}\;.
\end{eqnarray*}
Let now $\phi(t,x',y')$ be a test function such that $\phi(t,x',y') \ge
V(t,x',y')$ and $\phi(0,x,y) = V(0,x,y)$. Then there is a strategy $(u,a)$,
such that $V(0,x,y) < V^{u,a}(0,x,y) + h^2$. Choose a localisation sequence
$\{t_n\}$, such that 
\begin{eqnarray*}
&&\int_0^{\tau \wedge t_n \wedge t} [\sigma(s,Y_s,u_s) + a_s
\sigma_1(s,Y_s)] \phi_x(s,X_s^{u,a}, Y_s) \id W_s^1\;,\\  
&&\int_0^{\tau \wedge t_n \wedge t} a_s \sigma_2(s,Y_s) \phi_x(s,X_s^{u,a}, Y_s)\id
W_s^2\;,\\ 
\noalign{and} 
&&\int_0^{\tau\wedge t_n \wedge t} \sigma_Y(s,Y_s) \phi_y(s,X_s^{u,a},Y_s) \id W_s^Y
\end{eqnarray*}
are martingales, where as above, $\tau = \tau^{u, a} \wedge h$.
We have
\begin{eqnarray*}
\lefteqn{\phi(0,x,y) = V(0,x,y) \le V^{u,a}(0,x,y) + h^2}\hskip1cm\\ 
&=& \E[V(\tau\wedge t_n , X_{\tau \wedge t_n}, Y_{\tau \wedge t_n})] +
h^2 \le \E[\phi(\tau \wedge t_n, X_{\tau \wedge t_n}, Y_{\tau \wedge t_n})] +
h^2\\
&=& \phi(0,x,y) + \E\Bigl[ \int_0^{\tau \wedge t_n} [\phi_t(t,X_t, Y_t)\\ 
&&{}+ \{m(t, Y_t, u_t) + a_t \mu(t,Y_t)\} \phi_x(t,X_t, Y_t)\\ 
&&{}+ \onehalf\{ (\sigma(t,Y_t,u_t) + a_t \sigma_1(t,Y_t))^2 + a_t^2
\sigma_2^2(t, Y_t)\} \phi_{x x}(t,X_t, Y_t)\\ 
&&{}+  \mu_Y(t,Y_t) \phi_y(t,X_t,Y_t) + \onehalf \sigma_Y^2(t,Y_t) \phi_{y
y}(t,X_t,Y_t)] \id t\Bigr] + h^2\;.
\end{eqnarray*}
Because we consider a compact interval, we can let $n \to \infty$ and obtain
by bounded convergence
\begin{eqnarray*}
0 &\le& \E\Bigl[ \int_0^\tau [\phi_t(t,X_t, Y_t)\\ 
&&{}+ \{m(t, Y_t, u_t) + a_t \mu(t,Y_t)\} \phi_x(t,X_t, Y_t)\\ 
&&{}+ \onehalf\{ (\sigma(t,Y_t,u_t) + a_t \sigma_1(t,Y_t))^2 + a_t^2
\sigma_2^2(t, Y_t)\} \phi_{x x}(t,X_t, Y_t)\\ 
&&{}+  \mu_Y(t,Y_t) \phi_y(t,X_t,Y_t) + \onehalf \sigma_Y^2(t,Y_t) \phi_{y
y}(t,X_t,Y_t)] \id t\Bigr] + h^2\\
&\le& \E\Bigl[ \int_0^\tau \sup_{\bar u, \bar a}[\phi_t(t,X_t, Y_t)\\ 
&&{}+ \{m(t, Y_t, \bar u) + \bar a \mu(t,Y_t)\} \phi_x(t,X_t, Y_t)\\ 
&&{}+ \onehalf\{ (\sigma(t,Y_t,\bar u) + \bar a \sigma_1(t,Y_t))^2 + \bar a^2
\sigma_2^2(t, Y_t)\} \phi_{x x}(t,X_t, Y_t)\\ 
&&{}+  \mu_Y(t,Y_t) \phi_y(t,X_t,Y_t) + \onehalf \sigma_Y^2(t,Y_t) \phi_{y
y}(t,X_t,Y_t)] \id t\Bigr] + h^2\;.
\end{eqnarray*}
This gives by dividing by $h$ and letting $h \to 0$
\begin{eqnarray*}
0 &\le& \phi_t + \sup_{u,a} \{m(t,y,u) + a \mu(t,y)\} \phi_x\nonumber\\ 
&&{}+ \onehalf \{(\sigma(t,y, u) + a \sigma_1(t,y))^2 +  a^2 \sigma_2^2(t,y)\}
\phi_{x x} + \mu_Y(t,y) \phi_y\nonumber\\ 
&&{}+ \onehalf \sigma_Y^2(t,y) \phi_{y y}\;.
\end{eqnarray*}
This proves the assertion.
\end{proof}

Let now $u^*(t,x,y)$ and $a^*(t,x,y)$ be the maximiser in
\eqref{eq:hjb}. By \cite[Sec.~7]{wagner} we can choose these maximisers in a
measurable way. We further denote by
$u_t^* = u^*(t,X_t^{u^*, a^*}, Y_t)$ and $a_t^* = a^*(t,X_t^{u^*, a^*}, Y_t)$
the feedback strategy. 
\begin{theo}\label{thm:opt.strat}
Suppose that $V$ is a classical solution to the HJB equation
\eqref{eq:hjb}. Suppose further that the strategy $(u^*, a^*)$ admits a unique
strong solution for $X^{u^*, a^*}$ and that $\{X_t^{u^*, a^*}\}$ is uniformly
integrable. Then the strategy $(u^*, a^*)$ is optimal.
\end{theo}
\begin{proof}
By It\^o's formula we get for $X_t = X_t^{u^*,a^*}$
\begin{eqnarray*}
\lefteqn{V(t, X_t,Y_t) = V(0,x,y) +
\int_0^t [V_t(s,X_s, Y_s)}\hskip1.5cm\\
&&{}+ \{m(s, Y_s, u_s^*) + a_s^* \mu(s,Y_s)\} V_x(s,X_s, Y_s)\\ 
&&{}+ \onehalf\{ (\sigma(s,Y_s,u_s^*) + a_s^* \sigma_1(s,Y_s))^2 + {a_s^*}^2
\sigma_2^2(s, Y_s)\} V_{x x}(s,X_s, Y_s)\\ 
&&{}+  \mu_Y(s,Y_s) V_y(s,X_s,Y_s) + \onehalf \sigma_Y^2(s,Y_s) V_{y
y}(s,X_s,Y_s)] \id s\\ 
&&{}+ \int_0^t [\sigma(s,Y_s,u_s^*) + a_s^*
\sigma_1(s,Y_s)] V_x(s,X_s, Y_s) \id W_s^1\\ 
&&{}+ \int_0^t a_s \sigma_2(s,Y_s) V_x(s,X_s, Y_s)\id W_s^2\\
&&{}+ \int_0^t \sigma_Y(s,Y_s) V_y(s,X_s,Y_s) \id W_s^Y\\
&=& V(0,x,y) + \int_0^t a_s \sigma_2(s,Y_s) V_x(s,X_s, Y_s)\id W_s^2\\ 
&&{}+\int_0^t [\sigma(s,Y_s,u_s^*) + a_s^*
\sigma_1(s,Y_s)] V_x(s,X_s, Y_s) \id W_s^1\\
&&{}+ \int_0^t \sigma_Y(s,Y_s) V_y(s,X_s,Y_s) \id W_s^Y\;.
\end{eqnarray*}
Thus, $\{V(t, X_t,Y_t)\}$ is a local martingale. From $U(X_T) \le U(x) +
(X_T-x) U'(x)$ and the uniform integrability we get that $\{V(t, X_t,Y_t)\}$
is a martingale. We therefore have
\[ \E[U(X_T)] = \E[V(T, X_T, Y_T)] = V(0,x,y)\;.\]
This shows that the strategy is optimal.
\end{proof}
\begin{corollary}\label{cor:op.str}
Suppose that $V$ is a classical solution to the HJB equation
\eqref{eq:hjb}. Suppose further that the strategy $(u^*, a^*)$ admits a unique
strong solution for $X^{u^*, a^*}$ and that $\E[\int_0^T (a_t^*)^2 \id t] <
\infty$. Then the strategy $(u^*,a^*)$ is optimal.
\end{corollary}
\begin{proof}
Since the parameters are bounded, the condition $\E[\int_0^T (a_t^*)^2 \id t] <
\infty$ implies uniform integrability of $\{X_t^{u^*,a^*}\}$. The result
follows from Theorem~\ref{thm:opt.strat}.
\end{proof}


\section{SAHARA utility functions}
\label{section:sahara}

In this section we study the optimal reinsurance-investment problem when the
insurer's preferences are described by SAHARA utility functions. This class of
utility functions was first introduced by~\cite{chen:sahara} and it includes
the well known exponential and power utility functions as limiting cases. The
main feature is that SAHARA utility functions are well defined on the whole
real line and, in general, the risk aversion is non monotone.

More formally, we recall that a utility function $U\colon\R \to\R $ is of the SAHARA class if its absolute risk aversion (ARA) function $A(x)$ admits the following representation:
\begin{equation}
\label{eqn:arafun}
-\frac{U''(x)}{U'(x)} =: A(x)=\frac{a}{\sqrt{b^2+(x-d)^2}}\;, 
\end{equation}
where $a>0$ is the risk aversion parameter, $b>0$ the scale parameter and
$d\in\R $ the threshold wealth.

Let us try the ansatz
\begin{equation}
\label{eqn:ansatz_sahara}
V(t,x,y)=U(x)\tilde{V}(t,y)\;.
\end{equation}

\begin{remark}
By~\eqref{eqn:a_optimal} and~\eqref{eqn:ansatz_sahara}, the optimal investment strategy admits a simpler expression:
\begin{equation}
\label{eqn:a_sahara}
a^*(t,x,y) = \frac{\mu(t,y)-A(x)\sigma(t,y,u)\sigma_1(t,y)}{A(x)
(\sigma_1(t,y)^2+\sigma_2(t,y)^2)}\;.
\end{equation}
In particular, $a^*(t,x,y)$ is bounded by a linear function in $x$ and therefore our
assumption $\E[\int_0^T (a_t^*)^2 \id t] < \infty$ is fulfilled. Under our
hypotheses, if the HJB equation admits a classical solution, the assumptions
in Corollary~\ref{cor:op.str} are satisfied.
Let us note that $a^*(t,x,y)$ is influenced by the reinsurance strategy $u$.
\end{remark}


In this case~\eqref{eqn:hjb2} reads as follows:
\begin{eqnarray*}
0 &=& U(x)\tilde{V}_t+ \mu_Y(t,y) U(x)\tilde{V}_y
+ \onehalf \sigma_Y^2(t,y) U(x)\tilde{V}_{yy} \\
&&{}+ U'(x)\tilde{V}(t,y)\sup_{u\in[0,I]}\Psi_{t,x,y}(u)\;,
\end{eqnarray*}
where
\begin{multline}
\label{eqn:psi_sahara}
\Psi_{t,x,y}(u)\doteq m(t,y,u)\\
+\frac{\mu(t,y)^2-2\mu(t,y)\sigma(t,y,u)\sigma_1(t,y)A(x)-\sigma(t,y,u)^2\sigma_2(t,y)^2A(x)^2}{2[\sigma_1(t,y)^2+\sigma_2(t,y)^2]A(x)}\;.
\end{multline}
By our assumptions, $\Psi_{t,x,y}(u)$ is continuous in $u$, hence it admits a
maximum in the compact set $[0,I]$. However, we need additional requirements to
guarantee the uniqueness.

\begin{lemma}
\label{lemma:sahara_concavity}
If $m(t,y,u)$ is concave in $u\in[0,I]$ and $\sigma(t,y,u)$ is non negative and convex in $u\in[0,I]$, then there exists a unique maximiser for $\sup_{u\in[0,I]}\Psi_{t,x,y}(u)$.
\end{lemma}
\begin{proof}
We prove that $\Psi_{t,x,y}(u)$ is the sum of two concave functions, hence it is concave itself. As a consequence, there exists only one maximiser in $[0,I]$. Now, since $m(t,y,u)$ is strictly concave by hypothesis, we only need to show that
\[
\sigma(t,y,u)^2\sigma_2(t,y)^2A(x)+2\mu(t,y)\sigma(t,y,u)\sigma_1(t,y)
\]
is convex in $u$. We know that this quadratic form is convex and increasing when the argument is non negative. Recalling that $\sigma(t,y,u)\ge0$ by hypothesis, we can conclude that the function above is convex, because it is the composition of a non decreasing and convex function with a convex function ($\sigma(t,y,u)$ is so, by assumption). The proof is complete.
\end{proof}
\begin{remark}
Uniqueness is not necessary. If $u^*(t,x,y)$ is not unique, we have to choose
a measurable version in order to determine an optimal strategy.
\end{remark}


\subsection{Proportional reinsurance}

Let us consider the diffusion approximation to the classical risk model with
non-cheap proportional reinsurance, see e.g.{} \cite[Chapter 2]{schmidli:control}. More formally, 
\begin{equation}
\label{eqn:noncheap_model}
\id X^{0}_t=(p-q +qu)\id t + \sigma_0u\id W^1_t\;, \qquad X^{0}_0=x\;,
\end{equation}
with $p<q$ and $\sigma_0>0$. Here $I=1$.
From the economic point of view, the insurer transfers a proportion $1-u$ of her risks to the reinsurer (that is $u=0$ corresponds to full reinsurance).
In this case, by~\eqref{eqn:psi_sahara} our optimization problem reduces to
\begin{equation}
\label{eqn:noncheap_pb}
\sup_{u\in[0,1]} qu+\frac{\mu(t,y)^2-2\mu(t,y)\sigma_1(t,y)A(x)\sigma_0u-\sigma_2(t,y)^2A(x)^2\sigma_0^2u^2}{2[\sigma_1(t,y)^2+\sigma_2(t,y)^2]A(x)}\;.
\end{equation}

The optimal strategy is characterized by the following proposition.

\begin{prop}
\label{prop:sahara}
Under the model~\eqref{eqn:noncheap_model}, the optimal reinsurance-investment strategy is given by $(u^*(t,x,y),a^*(t,x,y))$, with
\begin{equation}
\label{eqn:u_sahara}
u^*(t,x,y)=
\begin{cases}
0 & \text{$(t,x,y)\in A_0$}\\
\frac{(\sigma_1(t,y)^2+\sigma_2(t,y)^2)q-\mu(t,y)\sigma_0\sigma_1(t,y)}{\sigma_0^2\sigma_2(t,y)^2A(x)}
& \text{$(t,x,y)\in (A_0\cup A_1)^C$} \\
1 & \text{$(t,x,y)\in A_1\;,$}
\end{cases}
\end{equation}
where
\begin{align*}
A_0&\doteq\Set{(t,x,y)\in [0,T]\times\R ^2 : q<\frac{\mu(t,y)\sigma_1(t,y)\sigma_0}{\sigma_1(t,y)^2+\sigma_2(t,y)^2}}\;,\\
A_1&\doteq\Set{(t,x,y)\in [0,T]\times\R ^2: q>\frac{\sigma_0[\sigma_2^2A(x)\sigma_0+\mu(t,y)\sigma_1(t,y)]}{\sigma_1(t,y)^2+\sigma_2(t,y)^2}}\;,
\end{align*}
and
\begin{equation}
\label{eqn:a_noncheap}
a^*(t,x,y)=\frac{\mu(t,y)-A(x)\sigma_0u^*(t,x,y)\sigma_1(t,y)}{A(x)(\sigma_1(t,y)^2+\sigma_2(t,y)^2)}\;.
\end{equation}
\end{prop}
\begin{proof}
The expression for $a^*(t,x,y)$ can be readily obtained by \eqref{eqn:a_sahara}.
By Lemma~\ref{lemma:sahara_concavity}, there exists a unique maximiser $u^*(t,x,y)$ for $\sup_{u\in[0,I]}\Psi_{t,x,y}(u)$, where $\Psi_{t,x,y}(u)$ is defined in~\eqref{eqn:psi_sahara} replacing $m(t,y,u)=p-(1-u)q$ and $\sigma(t,y,u)=\sigma_0u$. Now we notice that
\[
(t,x,y)\in A_0 \Rightarrow \frac{\partial \Psi_{t,x,y}(0)}{\partial u} <0\;,
\]
therefore full reinsurance is optimal. On the other hand,
\[
(t,x,y)\in A_1 \Rightarrow \frac{\partial \Psi_{t,x,y}(1)}{\partial u} >0\;,
\]
hence in this case null reinsurance is optimal. Now let us observe that
\[
(t,x,y)\in A_0 \Rightarrow q<\frac{\mu(t,y)\sigma_1(t,y)\sigma_0}{\sigma_1(t,y)^2+\sigma_2(t,y)^2}<\frac{\sigma_0[\sigma_2^2A(x)+\mu(t,y)\sigma_1(t,y)]}{\sigma_1(t,y)^2+\sigma_2(t,y)^2}\;,
\]
which implies $A_0\cap A_1=\emptyset$. Finally, when $(t,x,y)\in (A_0\cup A_1)^C$, the optimal strategy is given by the unique stationary point of $\Psi_{t,x,y}(u)$. By solving $\frac{\partial \Psi_{t,x,y}(u)}{\partial u}=0$, we obtain the expression in~\eqref{eqn:u_sahara}.
\end{proof}

\begin{remark}
The previous result holds true under the slight generalization of
$p(t,y)$, $q(t,y)$, $\sigma_0(t,y)$ dependent on time and on the environmental
process. In this case, there will be an additional effect of the exogenous factor $Y$.
\end{remark}
Proposition~\ref{prop:sahara} also holds in the case of an exponential utility
function.

\begin{corollary}
For $U(x)=-\e^{-\beta x}$ with $\beta>0$, the optimal strategy is given by
$(u^*(t,y),a^*(t,y))$, with
\begin{equation}
\label{eqn:u_exp}
u^*(t,y)=
\begin{cases}
0 & \text{$(t,y)\in A_0$}\\
\frac{(\sigma_1(t,y)^2+\sigma_2(t,y)^2)q-\mu(t,y)\sigma_0\sigma_1(t,y)}{\sigma_0^2\sigma_2(t,y)^2\beta}
& \text{$(t,y)\in (A_0\cup A_1)^C$} \\
1 & \text{$(t,y)\in A_1\;,$}
\end{cases}
\end{equation}
where
\begin{align*}
A_0&\doteq\Set{(t,y)\in [0,T]\times\R : q<\frac{\mu(t,y)\sigma_1(t,y)\sigma_0}{\sigma_1(t,y)^2+\sigma_2(t,y)^2}}\;,\\
A_1&\doteq\Set{(t,y)\in [0,T]\times\R : q>\frac{\sigma_0[\sigma_2^2\sigma_0\beta+\mu(t,y)\sigma_1(t,y)]}{\sigma_1(t,y)^2+\sigma_2(t,y)^2}}\;,
\end{align*}
and
\begin{equation}
\label{eqn:a_noncheap_exp}
a^*(t,y)=\frac{\mu(t,y)-\beta\sigma_0
u^*(t,y)\sigma_1(t,y)}{\beta(\sigma_1(t,y)^2+\sigma_2(t,y)^2)}\;.
\end{equation}
\end{corollary}
\begin{proof}
By definition of the ARA function, the exponential utility function corresponds to the special case $A(x)=\beta$. Hence, we can apply Proposition~\ref{prop:sahara}, by replacing the ARA function. All the calculations remain the same, but the optimal strategy will be independent on the current wealth level $x$.
\end{proof}


\subsection{Excess-of-loss reinsurance}

Now we consider the optimal excess-of-loss reinsurance problem. The retention
level is chosen in the interval $u\in[0,+\infty]$ and for any future claim 
the reinsurer is responsible for all the amount which exceeds that threshold $u$. 
For instance, $u = \infty$ corresponds to no reinsurance. 
The surplus process without investment is given
by, see also \cite{EisSchm}
\begin{equation}\label{eqn:XL_model}
\id X^{0}_t= \Bigl(\theta\int_0^u \bar{F}(z)\id z- (\theta-\eta)\E[Z] \Bigr)\id t
+ \sqrt{\int_0^u2z\bar{F}(z)\id z}\id W^1_t\;, \quad X^{0}_0=x\;,
\end{equation}
where $\theta,\eta>0$ are the reinsurer's and the insurer's safety loadings,
respectively, and $\bar{F}(z) = 1 - F(z)$ is the tail of the claim size
distribution function.
In the sequel we require $\E[Z] < \infty$ and, for the sake of the simplicity of the presentation, that
$F(z)<1$ $\forall z\in[0,+\infty)$.
Notice also that it is usually assumed $\theta>\eta$. 
However, we do not exclude the so called \textit{cheap reinsurance}, that is $\theta=\eta$.

By~\eqref{eqn:psi_sahara}, we obtain the following maximization problem:
\begin{eqnarray}
\label{eqn:XL_pb}
\lefteqn{\sup_{u\in[0,\infty]}\theta\int_0^u\bar{F}(z)\id z}\nonumber\\
&&{}-\frac{2\mu(t,y)\sigma_1(t,y)\sqrt{\int_0^u2z\bar{F}(z)\id z}+\sigma_2(t,y)^2
A(x)\int_0^u2z\bar{F}(z)\id z}{2[\sigma_1(t,y)^2+\sigma_2(t,y)^2]}\;.\hskip1cm
\end{eqnarray}

\begin{prop}
\label{prop:xl_sahara}
Under the model~\eqref{eqn:XL_model}, suppose that the function in
\eqref{eqn:XL_pb} is strictly concave in $u$. There exists a unique maximiser
$u^*(t,x,y)$ given by
\begin{equation}
\label{eqn:XL_solution}
u^*(t,x,y)=
\begin{cases}
	0 & \text{$(t,y)\in A_0$}
	\\
	\hat{u}(t,x,y) & \text{$(t,y)\in [0,T]\times\R \setminus A_0\;$}
\end{cases}
\end{equation}
where 
\[
A_0\doteq\Set{(t,y)\in[0,T]\times\R : \theta \le \frac{2\mu(t,y)\sigma_1(t,y)}{\sigma_1(t,y)^2+\sigma_2(t,y)^2}}
\]
and $\hat{u}(t,x,y)$ is the solution to the following equation:
\begin{equation}
\label{eqn:XL_nullderivative}
\theta(\sigma_1(t,y)^2+\sigma_2(t,y)^2)=2\mu(t,y)\sigma_1(t,y)\Bigl(\int_0^u2z\bar{F}(z)\id z\Bigr)^{-\frac{1}{2}}u+\sigma_2(t,y)^2A(x)u\;.
\end{equation}
\end{prop}
\begin{proof}
We first note that, using L'Hospital's rule,
\[ \lim_{u \to \infty} \frac{(\int_0^u \bar F(z) \id z)^2}{\int_0^u 2 z \bar
F(z) \id z} = 0\;.\]
The derivative with respect to $u$ of the function in \eqref{eqn:XL_pb} is
\[ \Bigl(\theta - u \frac{2\mu(t,y)\sigma_1(t,y) \bigl(\int_0^u 2z\bar{F}(z)\id
z\bigr)^{-\frac{1}{2}} +
\sigma_2(t,y)^2A(x)}{\sigma_1(t,y)^2+\sigma_2(t,y)^2}\Bigr) \bar F(u)\;.\]
Consider the expression between brackets
\begin{equation}\label{eq:XL_pb-der}
\theta - u \frac{2\mu(t,y)\sigma_1(t,y) \bigl(\int_0^u 2z\bar{F}(z)\id
z\bigr)^{-\frac{1}{2}} +
\sigma_2(t,y)^2A(x)}{\sigma_1(t,y)^2+\sigma_2(t,y)^2}\;.
\end{equation}
Since $\int_0^u 2 z \bar F(z) \id z \le u^2$, we see that
for any $(t,y)\in A_0$ the function in~\eqref{eqn:XL_pb} is strictly
decreasing. Thus $u^* = 0$ in this case. For $(t,y)\notin A_0$ we obtain by L'Hospital's rule,
\[ \lim_{u \to 0} \frac{\int_0^u 2 z \bar F(z) \id z}{u^2} = 1\;.\]
This implies that the function to be maximised increases close to zero. In
particular, the maximum is not taken at zero. Further, if $\E[Z^2] < \infty$,
then \eqref{eq:XL_pb-der} tends to $-\infty$ as $u \to \infty$. If
$\E[Z^2] = \infty$, then
\[ \lim_{u \to \infty} \frac{\int_0^u 2 z \bar F(z) \id z}{u^2} = 0\;.\]
Thus also in this case, \eqref{eq:XL_pb-der} tends to $-\infty$ as $u \to
\infty$. Thus the maximum is taken in $(0,\infty)$, and uniqueness of
$\hat{u}(t,x,y)$ is guaranteed by the concavity. Now the proof is complete.
\end{proof}

\begin{corollary}
Under the assumptions of Proposition~\ref{prop:xl_sahara}, the optimal re\-in\-sur\-ance-investment strategy is given by
\[ \biggl(\frac{\mu(t,y)-A(x)\sigma_1(t,y)\sqrt{\int_0^{u^*(t,x,y)} 2 z
\bar{F}(z) \id z}}{A(x)(\sigma_1(t,y)^2+\sigma_2(t,y)^2)},\;u^*(t,x,y)\biggr)\;,
\]
with $u^*(t,x,y)$ given in~\eqref{eqn:XL_solution}.
\end{corollary}
The main assumption of Proposition~\ref{prop:xl_sahara}, that is the concavity of the function in \eqref{eqn:XL_pb}, may be not easy to verify. In the next result we relax that hypothesis, only requiring the uniqueness of a solution to equation \eqref{eqn:XL_nullderivative}.

\begin{prop}
Under the model~\eqref{eqn:XL_model}, suppose that the equation \eqref{eqn:XL_nullderivative} admits a unique solution $\hat{u}(t,x,y)$ for any $(t,x,y)\in[0,T]\times\R ^2$. Then it is the unique maximiser to \eqref{eqn:XL_pb}.
\end{prop}
\begin{proof}
In the proof of Proposition~\ref{prop:xl_sahara} we only used the concavity to
verify uniqueness of the maximiser. Therefore, the same proof
applies.
\end{proof}


\subsection{Independent markets}

Suppose that the insurance and the financial markets are conditionally
independent given $Y$. That is, let $\sigma_1(t,x) = 0$. Then by~\eqref{eqn:a_sahara} we get
\[ a^*(t,x,y) = \frac{\mu(t,y)}{A(x) \sigma_2^2(t,y)}\;.\] 

\begin{remark}
Suppose that $\sigma(t,y,u)\ge0$ as usual. The insurer invests a larger amount of its surplus in the risky asset when the financial market is independent on the insurance market. Indeed, the reader can easily compare the formula above with \eqref{eqn:a_sahara}.
\end{remark}
Regarding the reinsurance problem, by \eqref{eqn:psi_sahara} we have to maximise this quantity:
\[ \Psi_{t,x,y}(u) := m(t,y,u)
+\frac{\mu(t,y)^2-\sigma(t,y,u)^2\sigma_2(t,y)^2A(x)^2}{2\sigma_2(t,y)^2
A(x)}\;.\]

\begin{prop}
Suppose that $\Psi_{t,x,y}(u)$ is strictly concave in $u\in[0,I]$. Then the optimal reinsurance strategy admits the following expression:
\begin{equation}
\label{eqn:independent_solution}
u^*(t,x,y)=
\begin{cases}
	0 & (t,x,y)\in A_0
	\\
	\hat{u}(t,x,y) & (t,x,y)\in [0,T]\times\R \setminus (A_0\cup A_I)
	\\
	I & (t,x,y)\in A_I\;,
\end{cases}
\end{equation}
where
\begin{align*}
A_0&\doteq\Set{(t,x,y)\in[0,T]\times\R ^2 : \pderiv{m(t,y,0)}u \le A(x)\sigma(t,y,0)\pderiv{\sigma(t,y,0)}u}\;,\\
A_I&\doteq\Set{(t,x,y)\in[0,T]\times\R ^2 : \pderiv{m(t,y,I)}u \ge A(x)\sigma(t,y,I) \pderiv{\sigma(t,y,I)}u}\;,
\end{align*}
and $\hat{u}(t,x,y)$ is the unique solution to
\[
\pderiv{m(t,y,u)}u = A(x) \sigma(t,y,u)\pderiv{\sigma(t,y,u)}u \;.
\]
\end{prop}
\begin{proof}
Since $\Psi_{t,x,y}(u)$ is continuous in $u$, it admits a unique maximiser in the compact set $[0,I]$. The derivative is
\[ \pderiv{m(t,y,u)}u - \onehalf A(x) \pderiv{\sigma^2(t,y,u)}u
=\pderiv{m(t,y,u)}u - A(x) \sigma(t,y,u)\pderiv{\sigma(t,y,u)}u \;.\]
If $(t,x,y)\in A_0$, then $\pderiv{\Psi_{t,x,y}(0)}u\le0$ and $\Psi_{t,x,y}(u)$ is decreasing in $[0,I]$, because it is concave; hence $u^*(t,x,y)=0$ is optimal $\forall(t,x,y)\in A_0$. Now notice that $A_0\cap A_I=\emptyset$, because of the concavity of $\Psi_{t,x,y}(u)$. If $(t,x,y)\in A_I$, then $\pderiv{\Psi_{t,x,y}(1)}u\ge0$ and $\Psi_{t,x,y}(u)$ is increasing in $[0,I]$, therefore it reaches the maximum in $u^*(t,x,y)=I$. Finally, if $(t,x,y)\in [0,T]\times\R \setminus (A_0\cup A_I)$, the maximiser coincides with the unique stationary point $\hat{u}(t,x,y)\in(0,I)$.
\end{proof}

The main consequence of the preceding result is that the reinsurance and the investment decisions depend on each other only via the surplus process and not via the parameters.

Now we specialize Propositions~\ref{prop:sahara} and~\ref{prop:xl_sahara} to the special case $\sigma_1(t,x) = 0$.

\begin{corollary}
\label{corollary:prop_nullsigma1}
Suppose that $\sigma_1(t,x) = 0$ and consider the case of proportional
reinsurance~\eqref{eqn:noncheap_model}. The optimal retention level is given by
\begin{equation*}
u^*(x)=\frac{q}{\sigma_0^2A(x)}\land1\;.
\end{equation*}
\end{corollary}
\begin{proof}
It is a direct consequence of the Propositions~\ref{prop:sahara}. In fact, the reader can easily verify that $A_0=\emptyset$ and the formula~\eqref{eqn:u_sahara} simplifies as above.
\end{proof}
As expected, the optimal retention level is proportional to the reinsurance
cost and inversely proportional to the risk aversion. Moreover, reinsurance 
is only bought for wealth not too far from $d$ (recall
equation~\eqref{eqn:arafun}). Note that the optimal strategy is independent on
$t$ and $y$, i.e. it is only affected by the current wealth.
Finally, full
reinsurance is never optimal. 


\begin{corollary}
Suppose that $\sigma_1(t,x) = 0$ and consider excess-of-loss
reinsurance~\eqref{eqn:XL_model}. The optimal retention level is given by
\begin{equation*}
u^*(x)=\frac{\theta}{A(x)}\;.
\end{equation*}
\end{corollary}
\begin{proof}
Using Proposition~\ref{prop:xl_sahara} we readily check that $A_0=\emptyset$ and by equation~\eqref{eqn:XL_nullderivative} we get the explicit solution $u^*(x)$.
\end{proof}
Again, the retention level turns out to increase with the reinsurance
safety loading and decrease with the risk aversion parameter. In addition,
it increases with the distance between the current wealth $x$ and the
threshold $d$.


\section{Numerical results}
\label{section:numerical}

In this section we provide some numerical examples based on Proposition \ref{prop:sahara}. All the simulations are performed according to the parameters in Table \ref{tab:parameters} below, unless indicated otherwise.

\begin{table}[H]
\caption{Simulation parameters}
\label{tab:parameters}
\centering
\begin{tabular}{ll}
\toprule
\textbf{Parameter} & \textbf{Value}\\
\midrule
$\mu$ & $0.08$\\
$\sigma_1$ & $0.5$\\
$\sigma_2$ & $0.5$\\
$\sigma_0$ & $0.5$\\
$q$ & $0.05$\\
$x$ & $1$\\
$a$ & $1$\\
$b$ & $1$\\
$d$ & $0$\\
\bottomrule
\end{tabular}
\end{table}

The choice of constant parameters may be considered as fixing
$(t,y,x)\in[0,T]\times\R^2$. Note that the strategy depends on $\{Y_t\}$ via
the parameters only. Now we illustrate how the strategy depends on the
different parameters. In the following figures, the solid line shows the
reinsurance strategy, the dashed line the investment strategy.

First, we analyse how the volatility coefficients of the risky asset influence
the optimal strategies. In Figures \ref{img:sigma1} and \ref{img:sigma2}
we notice very different behaviour.
On the one hand, the retention level $u^*$ is convex with respect to $\sigma_1$ up to a certain threshold, above which null reinsurance is optimal. On the other, when $\sigma_1>0$ (see Figure \ref{img:sigma2_a}) $u^*$ is null up to a given point and concave with respect to $\sigma_2$ from that point on. Finally, for $\sigma_1=0$ (see Figure \ref{img:sigma2_b}) the retention level is constant (see Corollary \ref{corollary:prop_nullsigma1}). Let us observe that the regularity of the optimal investment in Figure \ref{img:sigma2_b} is due to the absence of influence from $u^*$ (which remains constant).

\begin{figure}[H]
\centering
\ifpdf
\includegraphics[scale=0.28]{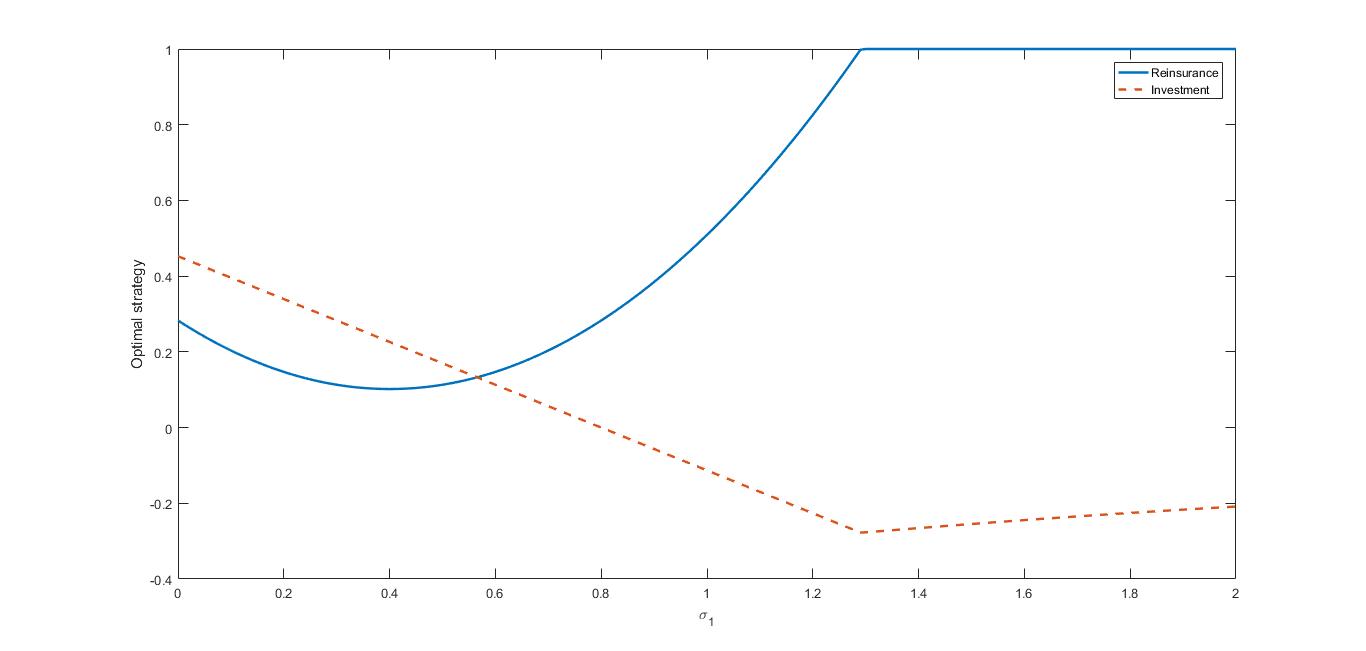}
\else \includegraphics[scale=0.28]{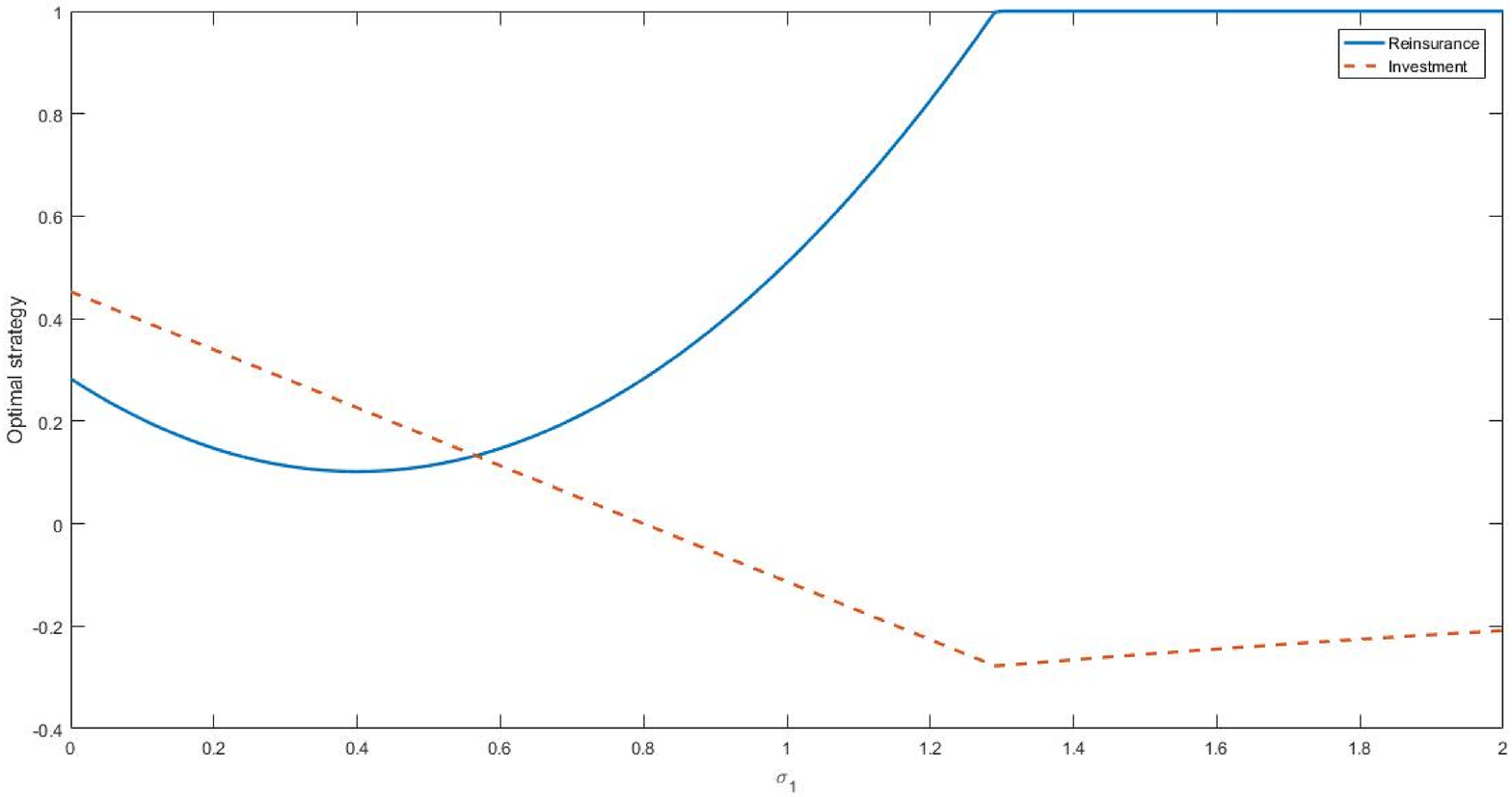}
\fi
\caption{The effect of $\sigma_1$ on the optimal reinsurance-investment strategy.}
\label{img:sigma1}
\end{figure}

\begin{figure}[H]
	\centering
	\begin{subfigure}{1\textwidth} 
                \ifpdf
		\includegraphics[width=\textwidth]{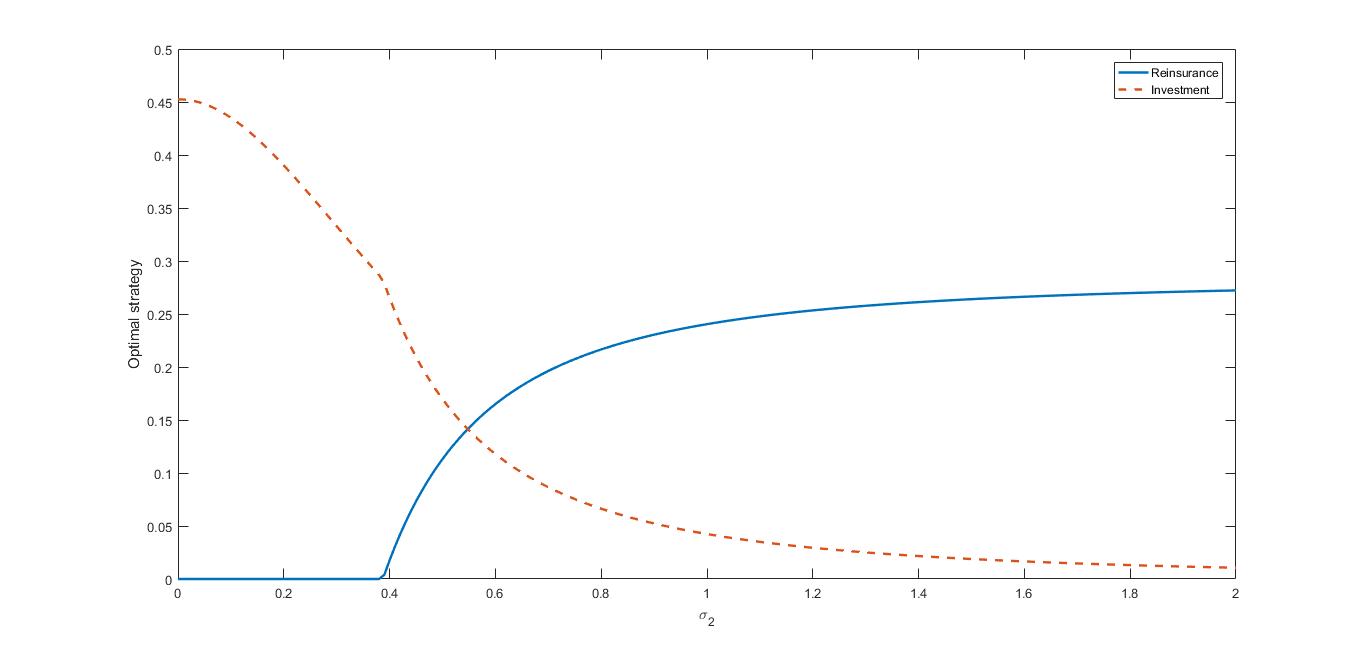}
	\else	\includegraphics[width=\textwidth]{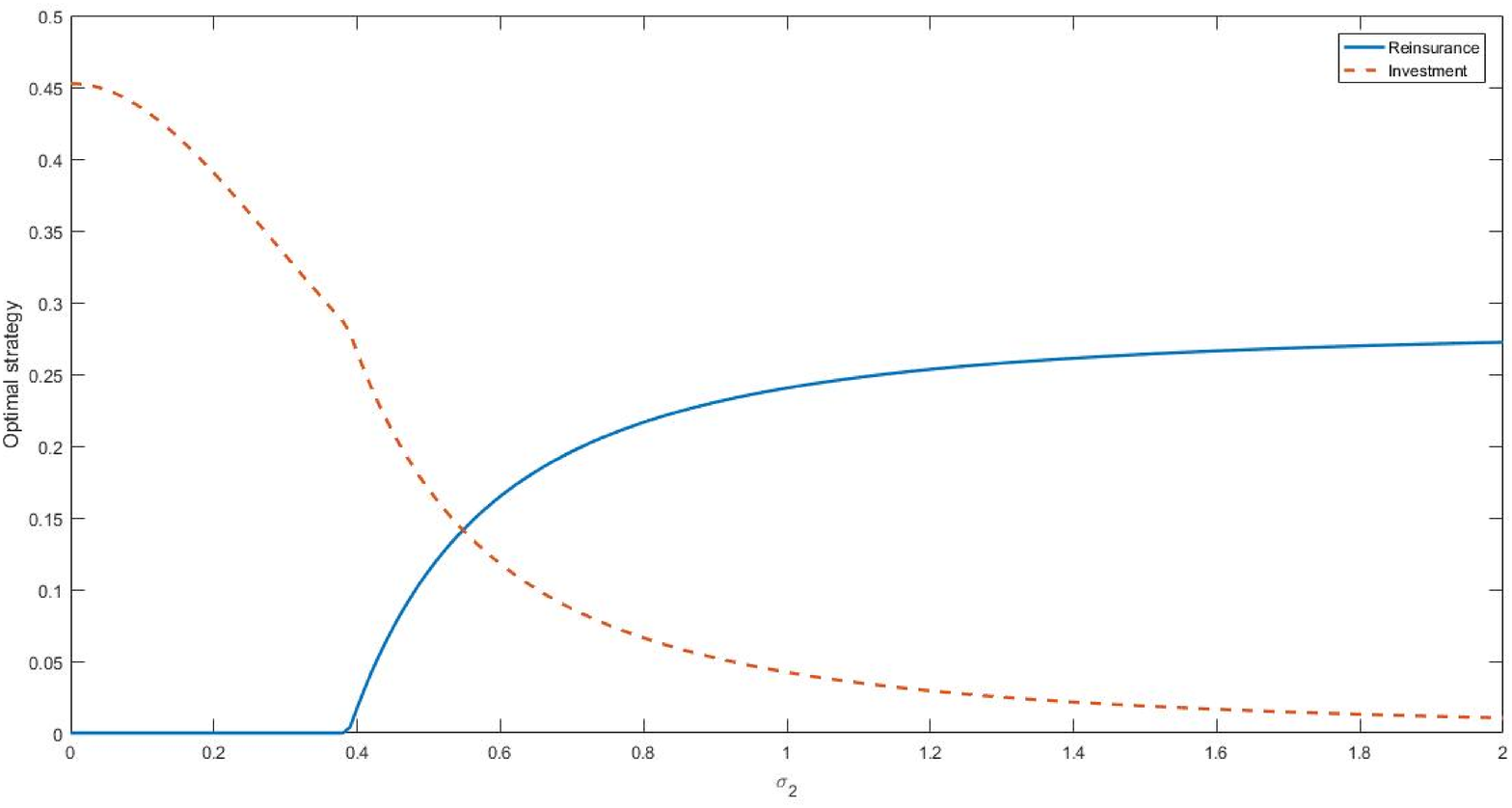}
                \fi
		\caption{Case $\sigma_1>0$} 
		\label{img:sigma2_a}
	\end{subfigure}
	\vspace{1em} 
	\begin{subfigure}{1\textwidth} 
                \ifpdf
		\includegraphics[width=\textwidth]{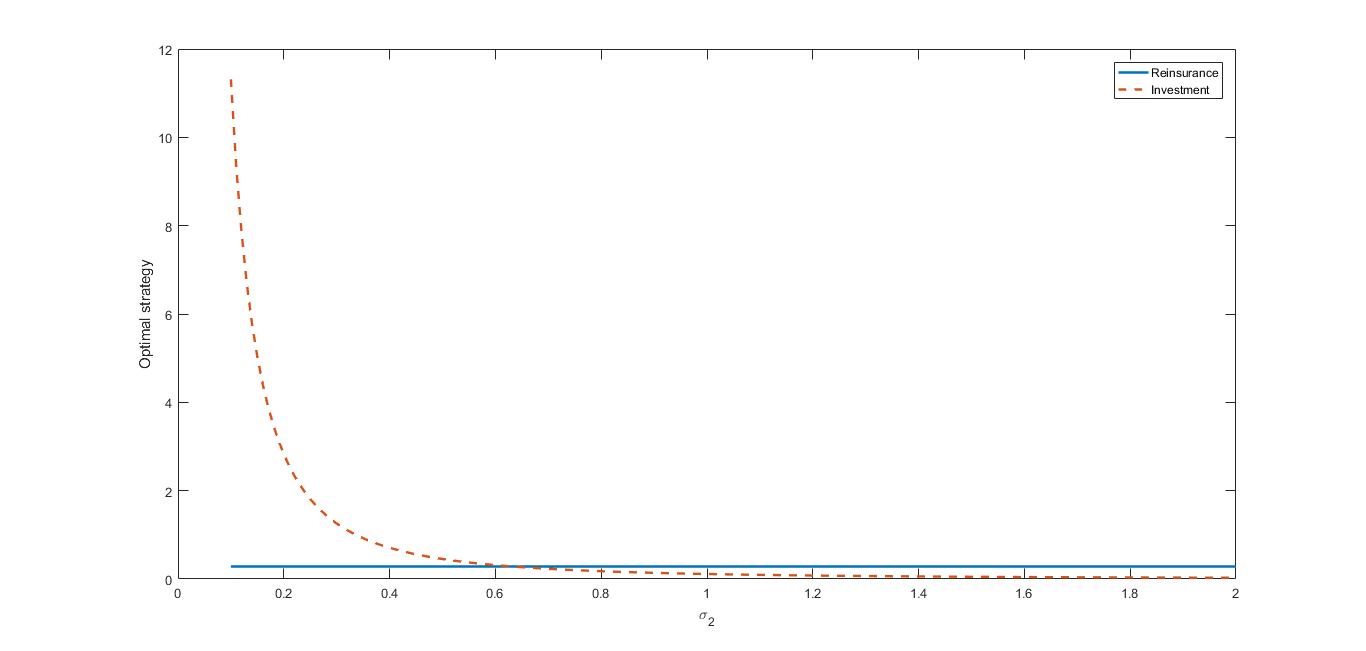}
	\else	\includegraphics[width=\textwidth]{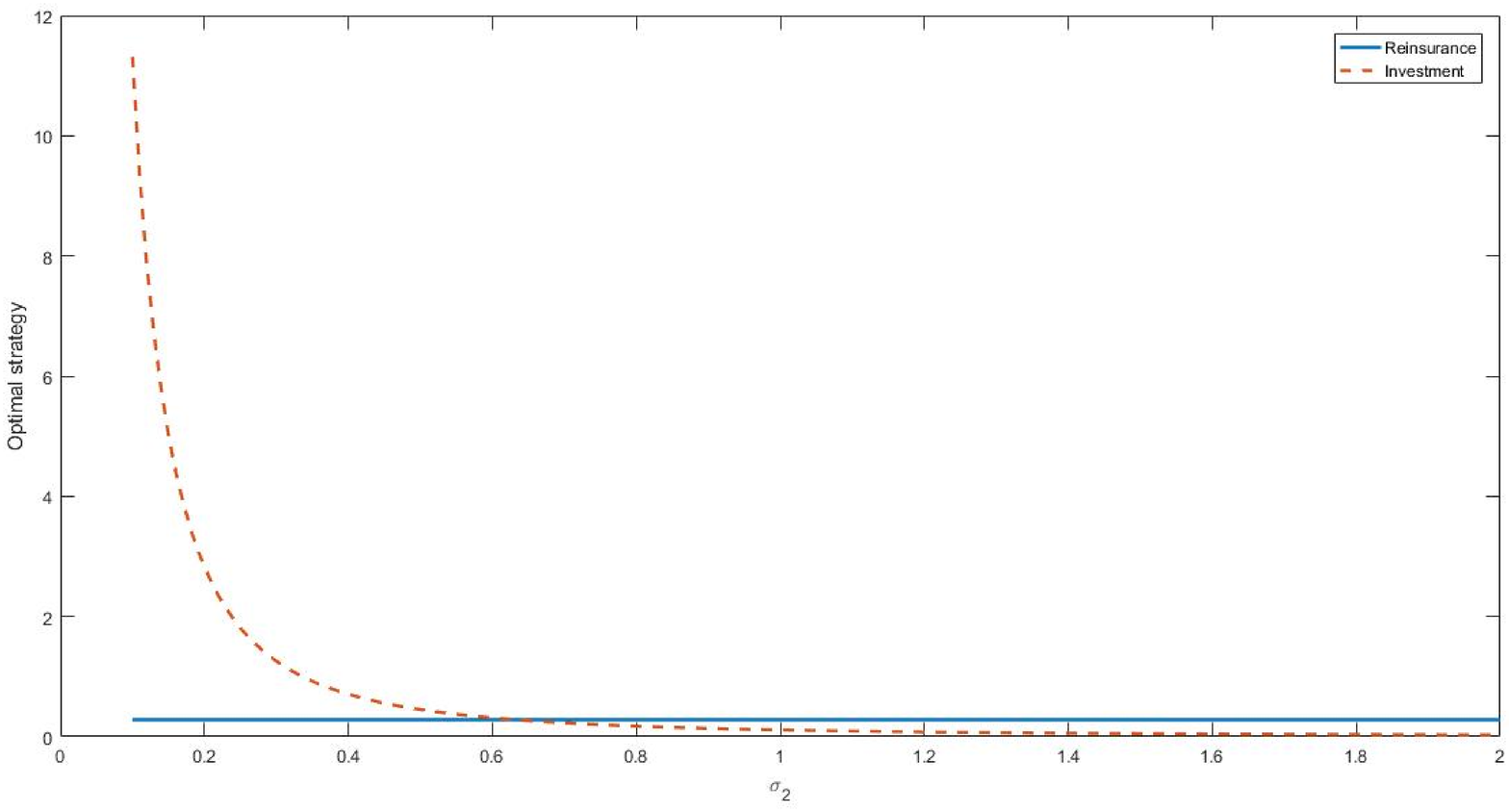}
                \fi
		\caption{Case $\sigma_1=0$} 
		\label{img:sigma2_b}
	\end{subfigure}
	\caption{The effect of $\sigma_2$ on the optimal reinsurance-investment strategy.} 
\label{img:sigma2}
\end{figure}


Now let us focus on Figure \ref{img:sigma0}. When $\sigma_0$ increases the insurer rapidly goes from null reinsurance to full reinsurance, while the investment $a^*$ strongly depends on the retention level $u^*$. Under $\sigma_1>0$ (see Figure \ref{img:sigma0_a}), as long as $u^*=1$, $a^*$ decreases with $\sigma_0$; when $u^*\in(0,1)$ starts decreasing, $a^*$ increases; finally, when $u^*$ stabilises at $0$, then $a^*$ stabilises at the starting level. On the contrary, when $\sigma_1=0$ the investment remains constant and $u^*$ asymptotically goes to $0$.

\begin{figure}[H]
	\centering
	\begin{subfigure}{1\textwidth} 
		\ifpdf
                \includegraphics[width=\textwidth]{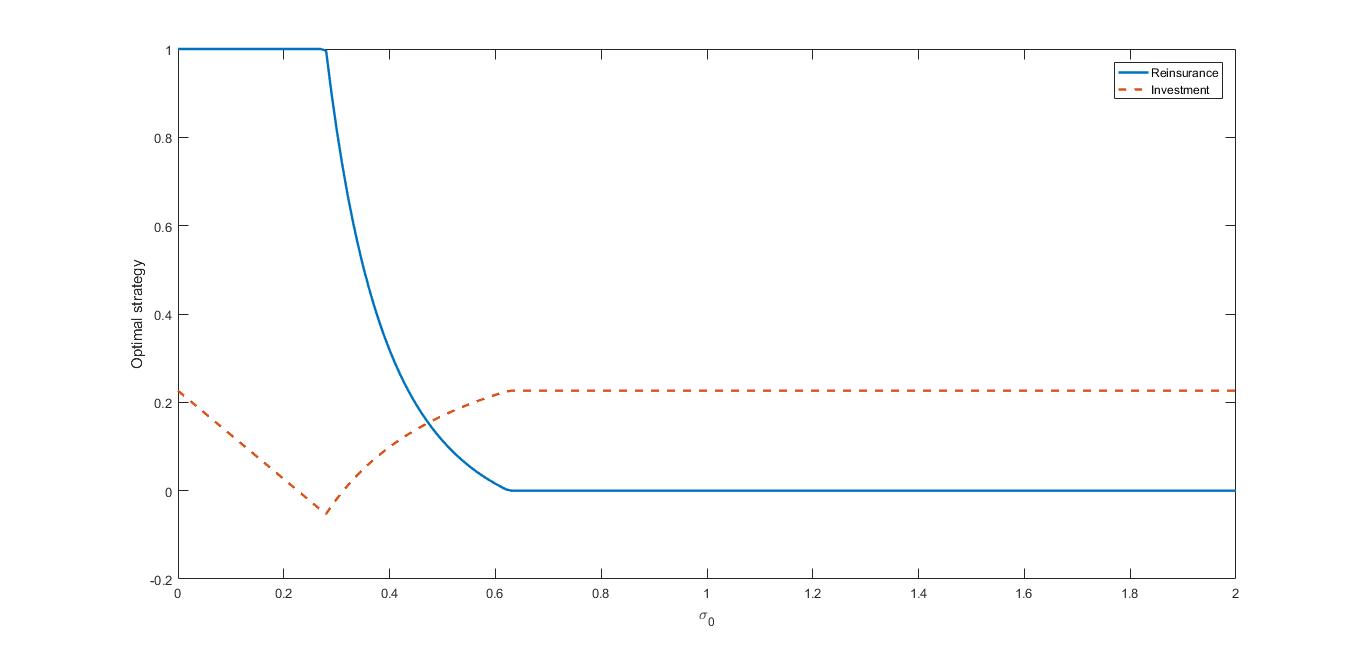}
         \else  \includegraphics[width=\textwidth]{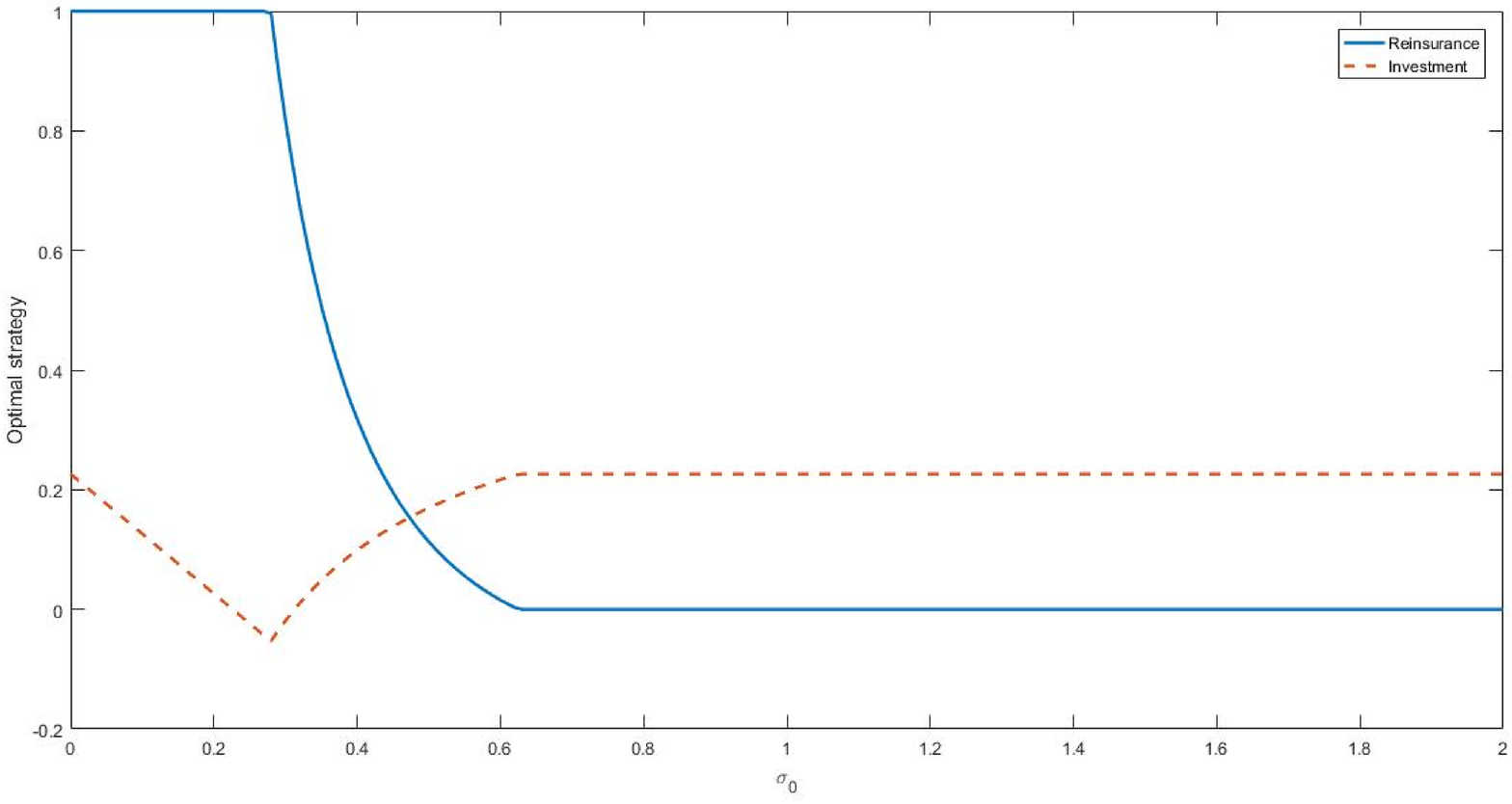}
                \fi
		\caption{Case $\sigma_1>0$} 
		\label{img:sigma0_a}
	\end{subfigure}
	\vspace{1em} 
	\begin{subfigure}{1\textwidth} 
		\ifpdf
                \includegraphics[width=\textwidth]{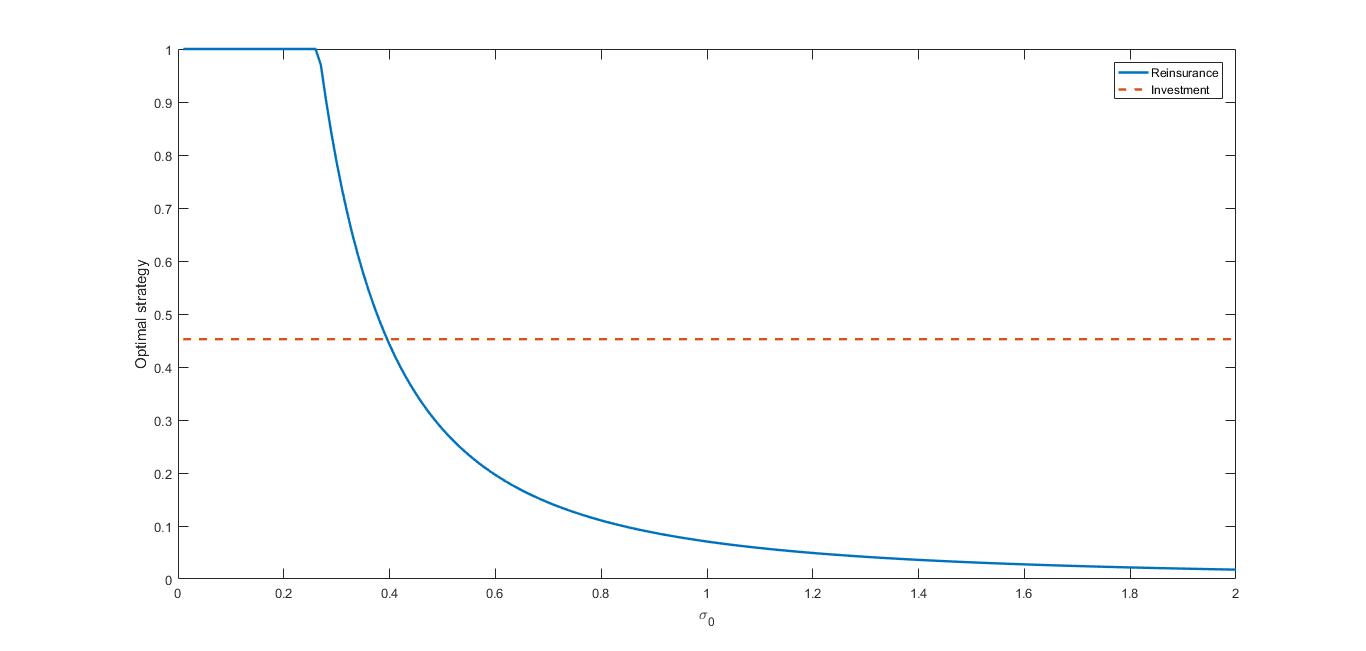}
         \else  \includegraphics[width=\textwidth]{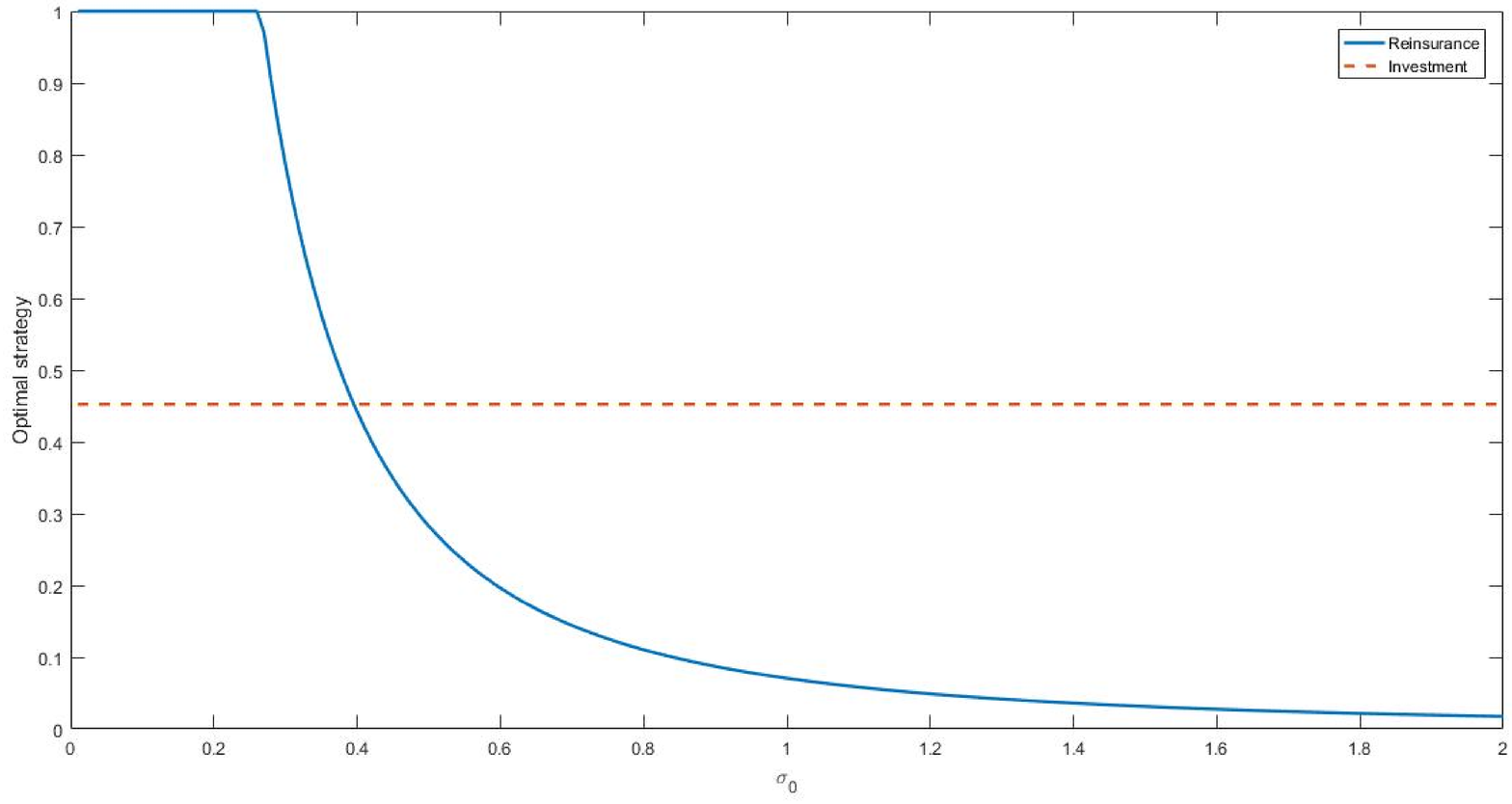}
                \fi
		\caption{Case $\sigma_1=0$} 
		\label{img:sigma0_b}
	\end{subfigure}
	\caption{The effect of $\sigma_0$ on the optimal reinsurance-investment strategy.} 
\label{img:sigma0}
\end{figure}


As pointed out in the previous section, the current wealth level $x$ plays an
important role in the evaluation of the optimal strategy and this is still
true under the special case $\sigma_0=0$. In Figure \ref{img:x} below we
illustrate the optimal strategy as a function of $x$. Both the reinsurance and
the investment strategy are symmetric with respect to $x=d=0$. Moreover, they
both increase when $x$ moves away from the threshold wealth $d$. This is not
surprising because the risk aversion decreases with the distance to
$d$.

\begin{figure}[H]
\centering
\ifpdf
\includegraphics[scale=0.28]{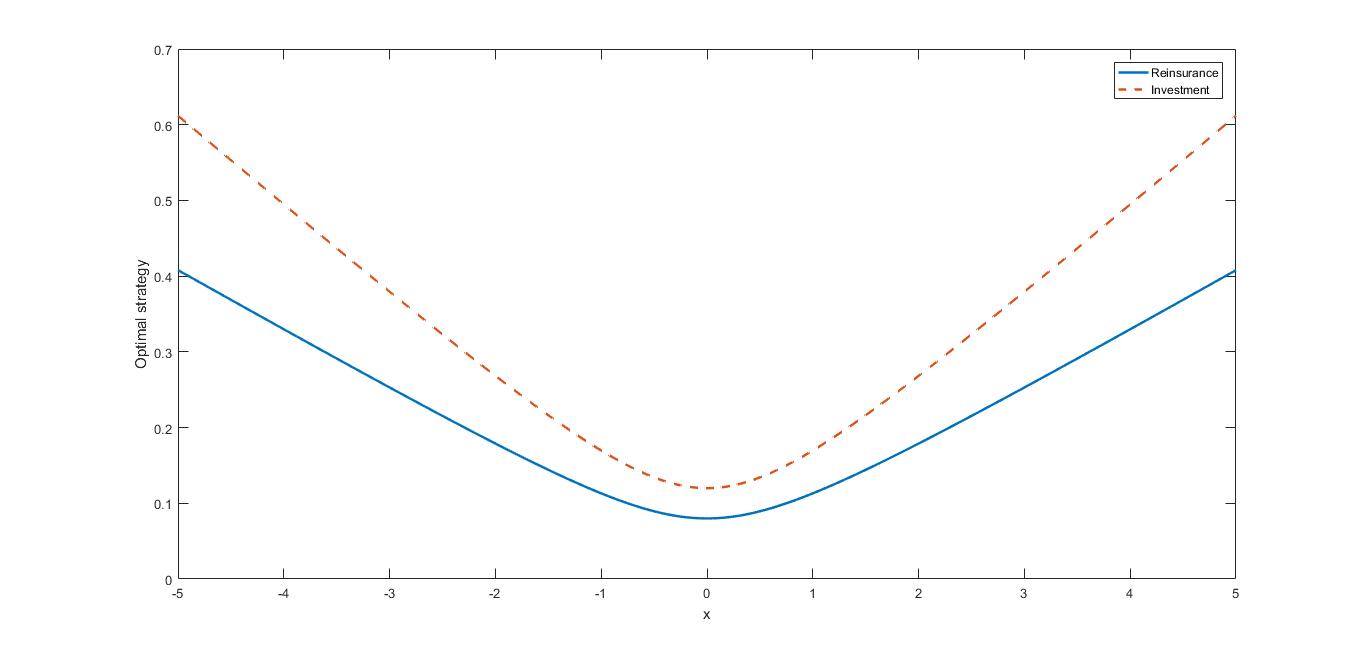}
\else\includegraphics[scale=0.28]{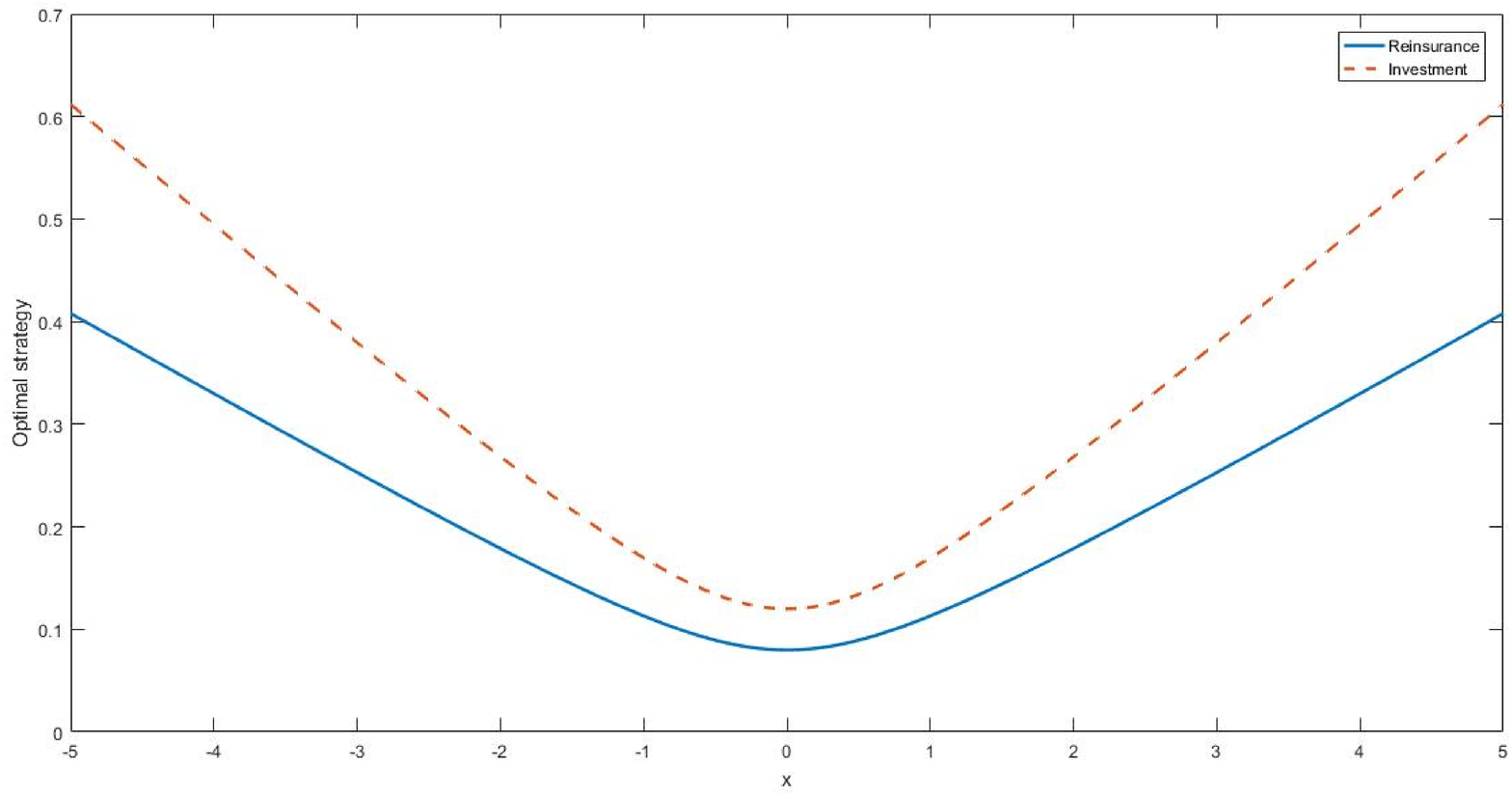}
\fi
\caption{The effect of the current wealth $x$ on the optimal reinsurance-investment strategy.}
\label{img:x}
\end{figure}

In the next Figure \ref{img:utility} we investigate the optimal strategy
reaction to modifications of the utility function. As expected, the higher is
the risk aversion, the larger is the optimal protection level and the lower is the investment in the risky asset (see Figure \ref{img:a}). When $b$ increases, both the investment and the retention level monotonically increase (see Figure \ref{img:b}). Let us recall that $b\to0$ corresponds to HARA utility functions. Finally, by Figure \ref{img:d} we notice that any change of $d$ produces the same result of a variation in current wealth $x$ (see Figure \ref{img:x}).

\begin{figure}[H]
	\centering
	\begin{subfigure}{1\textwidth} 
		\ifpdf
                \includegraphics[width=0.75\textwidth]{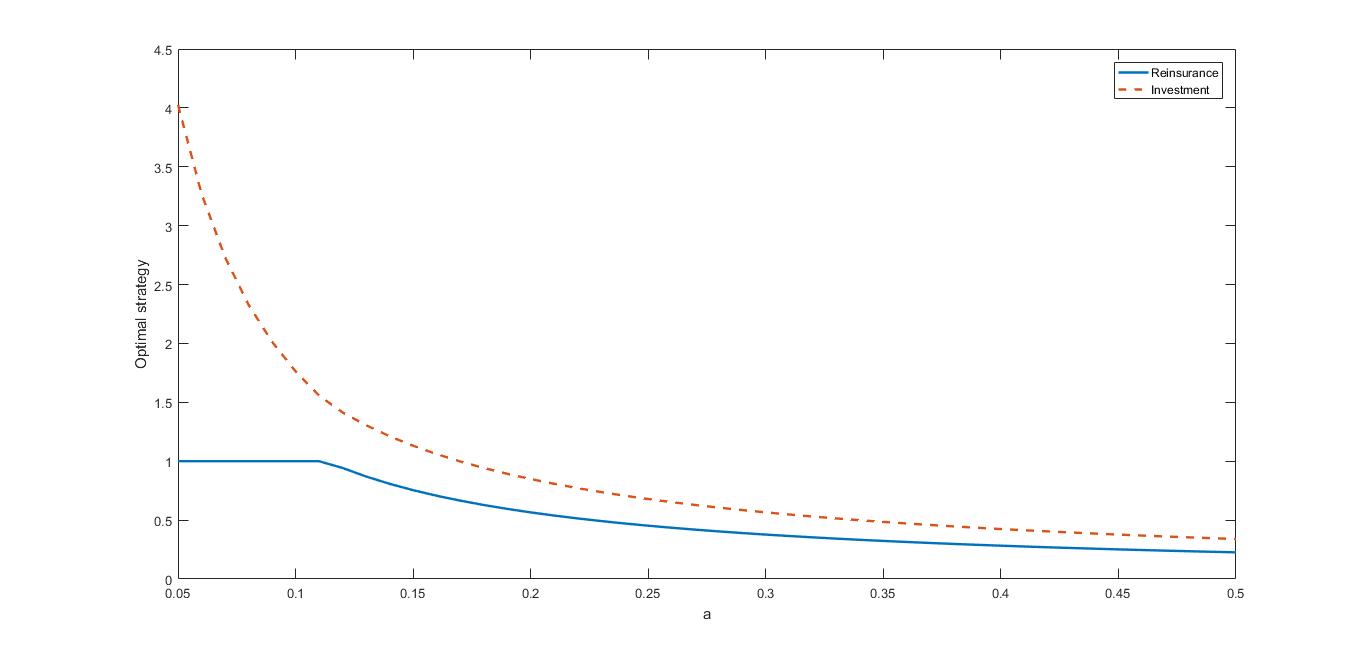}
        \else   \includegraphics[width=0.75\textwidth]{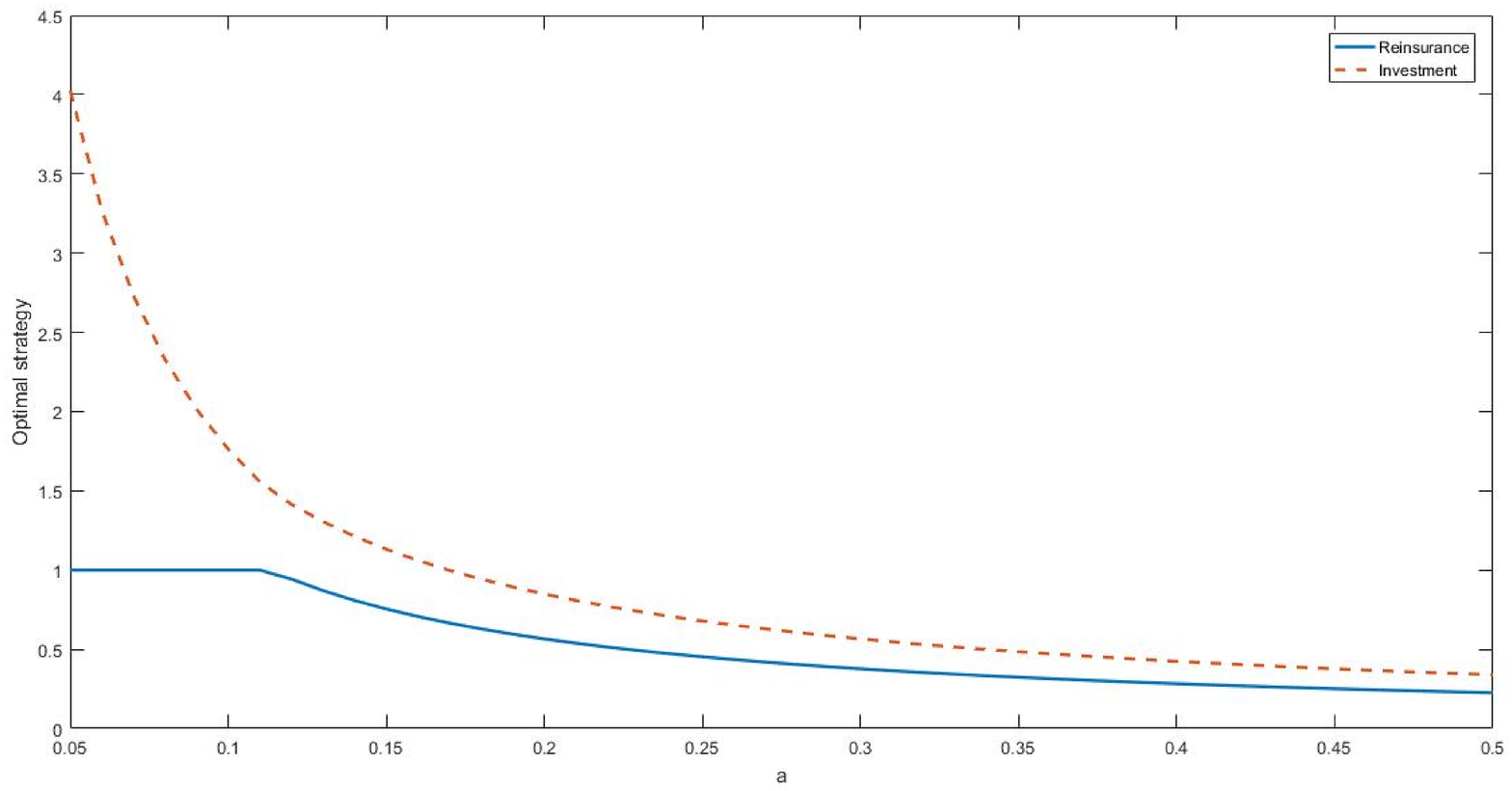}
                \fi
		\caption{The effect of the risk aversion on the optimal strategy.} 
		\label{img:a}
	\end{subfigure}
	\vspace{1em} 
	\begin{subfigure}{1\textwidth} 
		\ifpdf
                \includegraphics[width=0.75\textwidth]{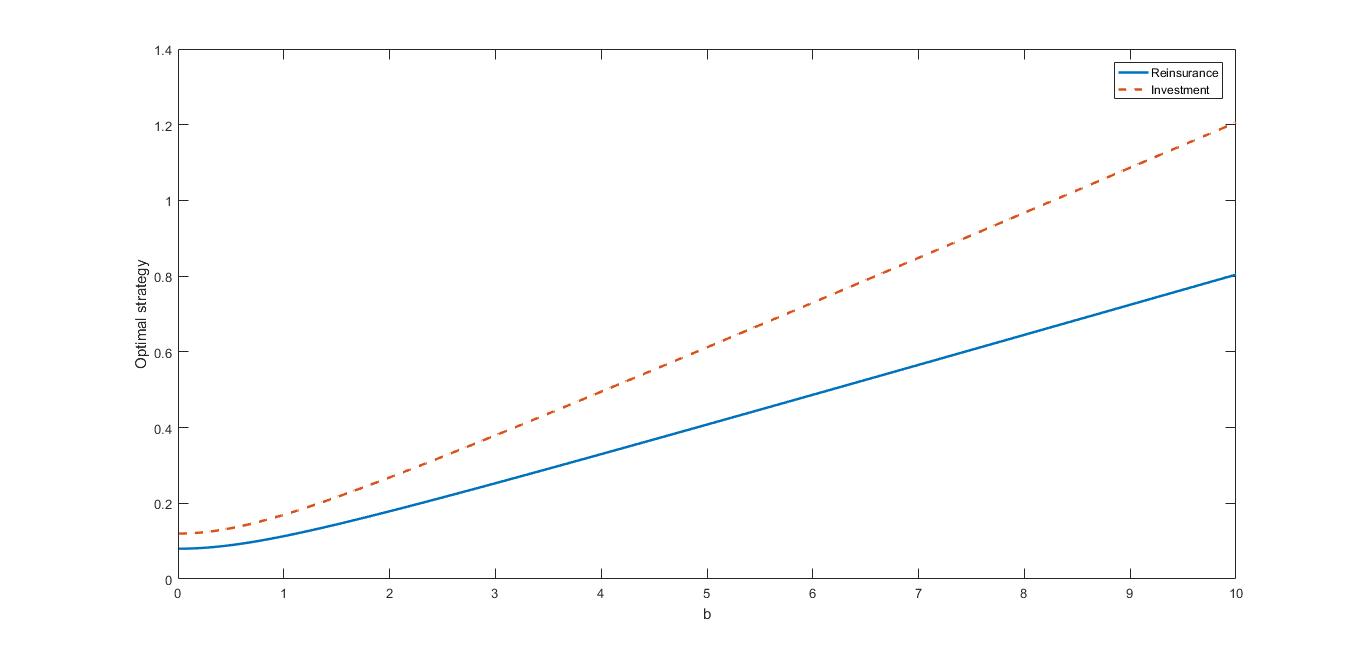}
        \else   \includegraphics[width=0.75\textwidth]{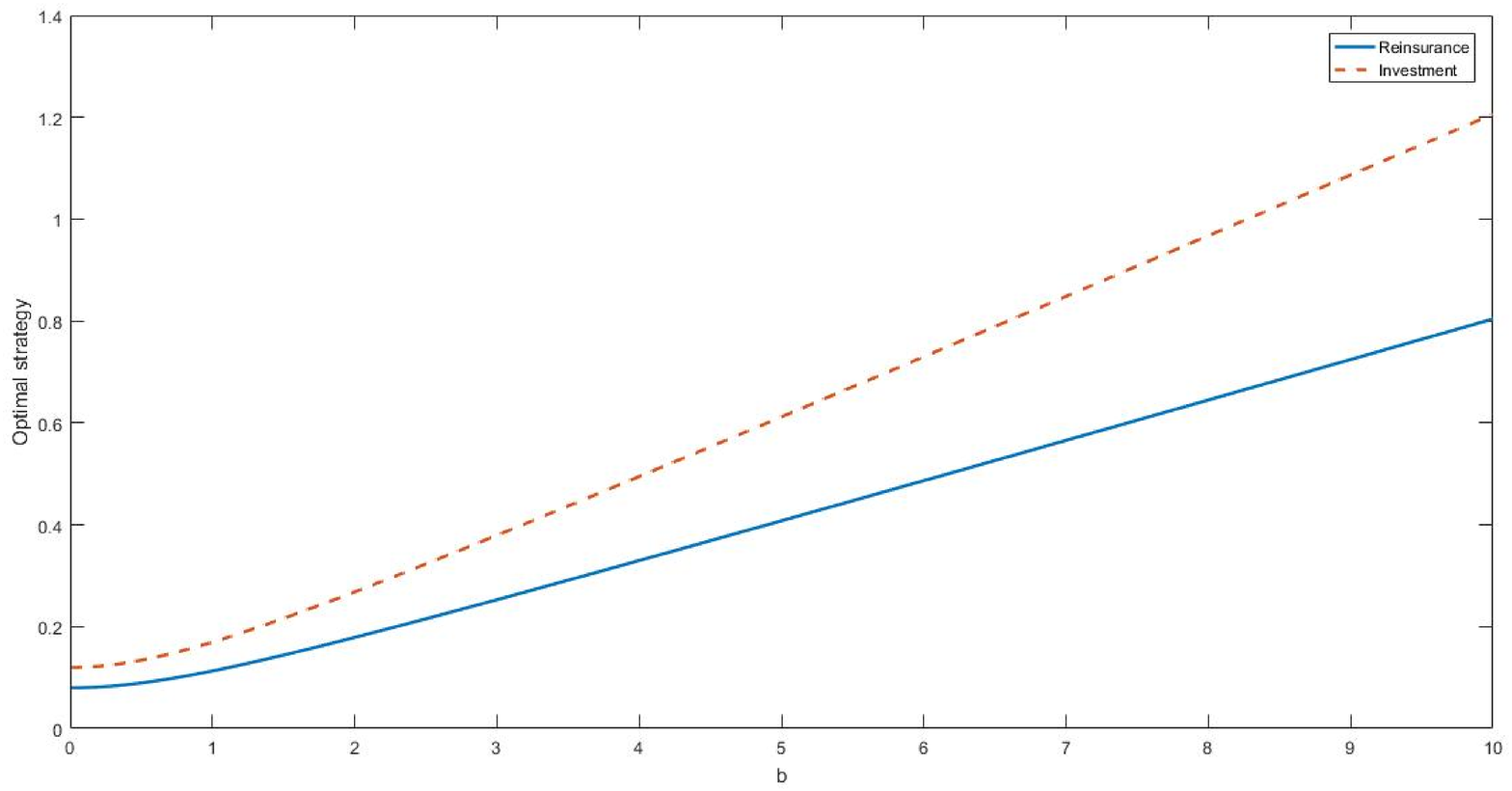}
                \fi
		\caption{The effect of the scale parameter on the optimal strategy.} 
		\label{img:b}
	\end{subfigure}
	\vspace{1em} 
	\begin{subfigure}{1\textwidth} 
                \ifpdf
		\includegraphics[width=0.75\textwidth]{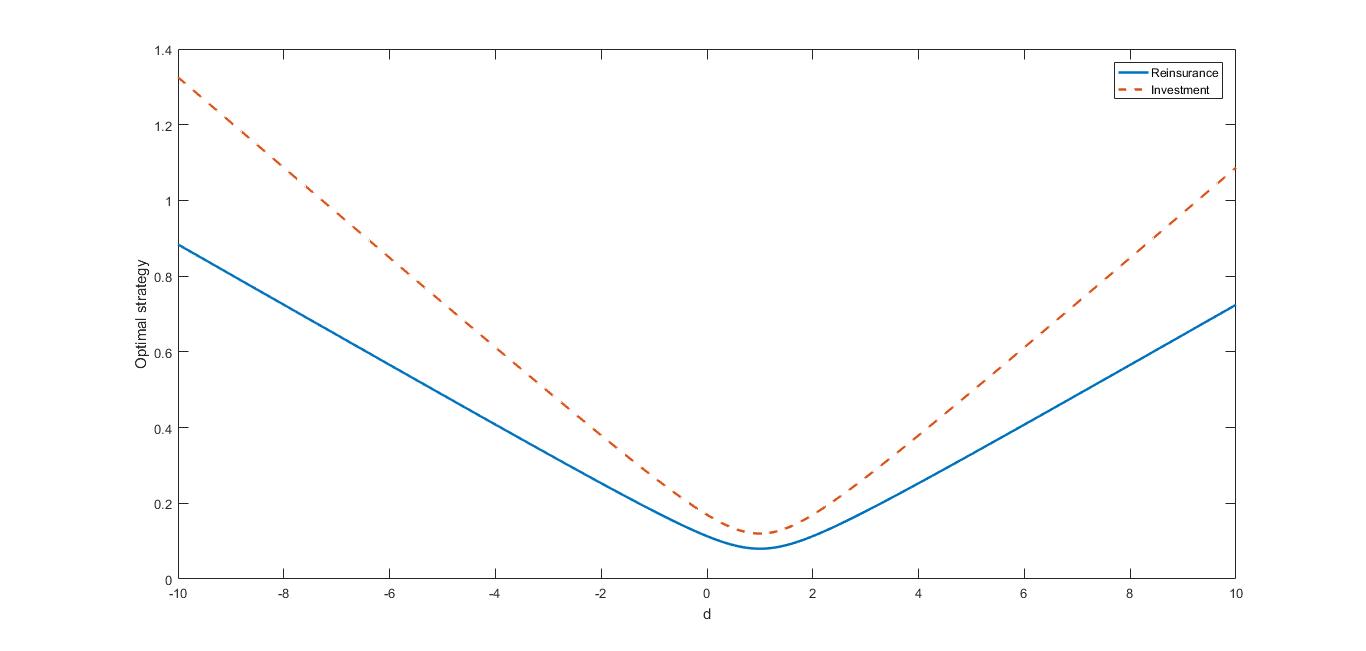}
	\else	\includegraphics[width=0.75\textwidth]{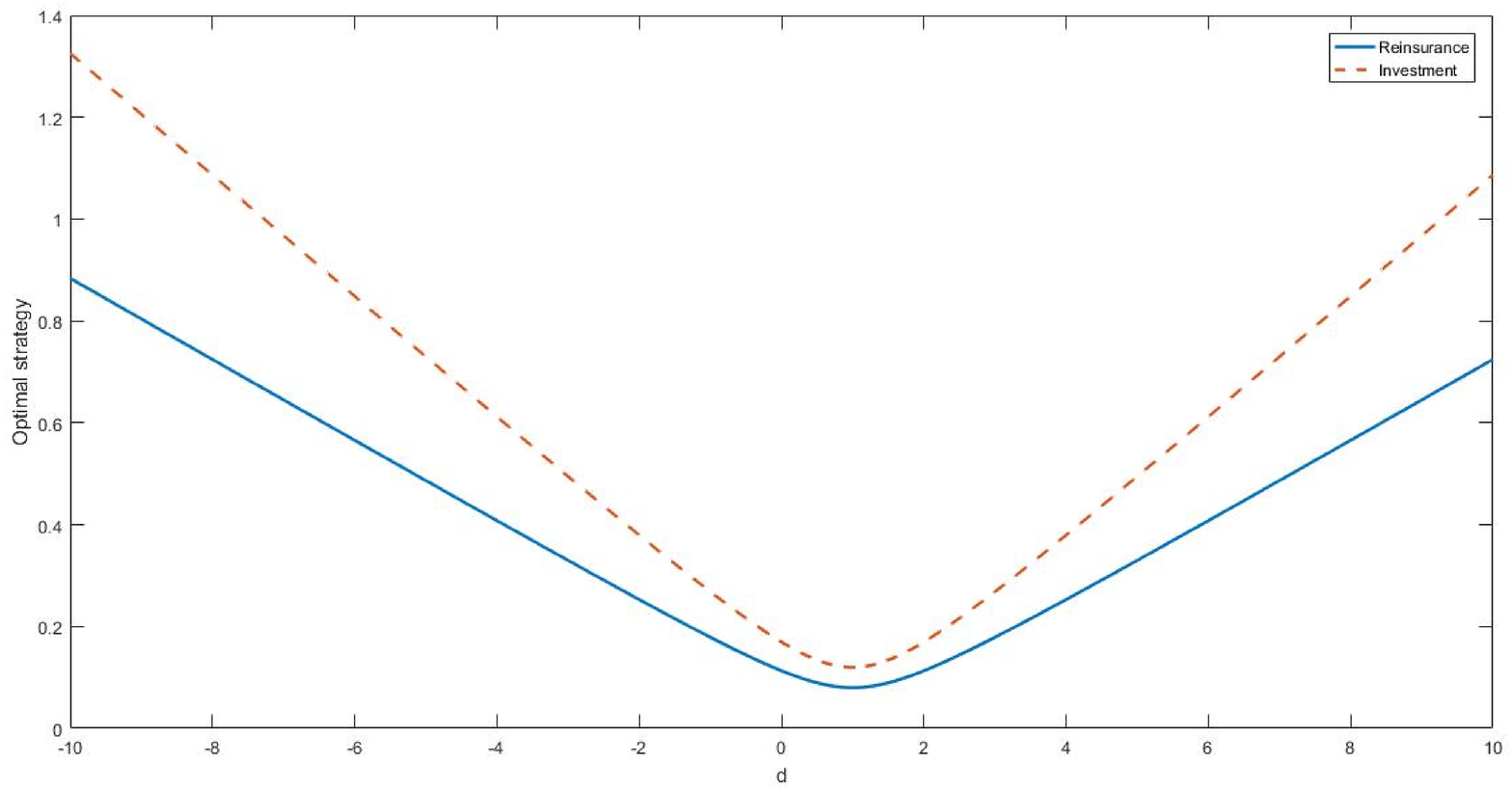}
                \fi
		\caption{The effect of the wealth threshold on the optimal strategy.} 
		\label{img:d}
	\end{subfigure}
	\caption{The effect of the SAHARA utility function parameters on the optimal reinsurance-investment strategy.} 
\label{img:utility}
\end{figure}


\newpage
\bibliographystyle{abbrv}
\bibliography{biblio}

\begin{thebibliography}{10}

\bibitem{BC:IME2019}
M.~Brachetta and C.~Ceci.
\newblock Optimal proportional reinsurance and investment for stochastic factor
  models.
\newblock {\em Insurance: Mathematics and Economics (in press)}, 2019.

\bibitem{buhlmann:1970}
H.~B\"uhlmann.
\newblock {\em Mathematical Methods in Risk Theory}.
\newblock Springer-Verlag, 1970.

\bibitem{chen:sahara}
A.~Chen, A.~Pelsser, and M.~Vellekoop.
\newblock Modeling non-monotone risk aversion using sahara utility functions.
\newblock {\em Journal of Economic Theory}, 146(5):2075--2092, 2011.

\bibitem{definetti:1940}
B.~{de Finetti}.
\newblock {Il problema dei ``pieni''.}
\newblock {\em {G. Ist. Ital. Attuari}}, 11:1--88, 1940.

\bibitem{EisSchm}
J.~Eisenberg and H.~Schmidli.
\newblock Optimal control of capital injections by reinsurance in a diffusion
  approximation.
\newblock {\em Blätter DGVFM}, 30:1--13, 2009.

\bibitem{gerber:1979}
H.~U. Gerber.
\newblock {\em An Introduction to Mathematical Risk Theory}.
\newblock Huebner Foundation Monographs, 1979.

\bibitem{grandell:risk}
J.~Grandell.
\newblock {\em Aspects of risk theory}.
\newblock Springer-Verlag, 1991.

\bibitem{GUERRA2008529}
M.~Guerra and M.~Centeno.
\newblock Optimal reinsurance policy: The adjustment coefficient and the
  expected utility criteria.
\newblock {\em Insurance: Mathematics and Economics}, 42(2):529 -- 539, 2008.

\bibitem{irgens_paulsen:optcontrol}
C.~Irgens and J.~Paulsen.
\newblock Optimal control of risk exposure, reinsurance and investments for
  insurance portfolios.
\newblock {\em Insurance: Mathematics and Economics}, 35:21--51, 2004.

\bibitem{liangbayraktar:optreins}
Z.~Liang and E.~Bayraktar.
\newblock Optimal reinsurance and investment with unobservable claim size and
  intensity.
\newblock {\em Insurance: Mathematics and Economics}, 55:156--166, 2014.

\bibitem{Mania2010}
M.~Mania and M.~Santacroce.
\newblock Exponential utility maximization under partial information.
\newblock {\em Finance and Stochastics}, 14(3):419--448, Sep 2010.

\bibitem{promislow2005}
S.~D. Promislow and V.~R. Young.
\newblock Minimizing the probability of ruin when claims follow brownian motion
  with drift.
\newblock {\em North American Actuarial Journal}, 9(3):109, 2005.

\bibitem{schmidli:2001}
H.~Schmidli.
\newblock Optimal proportional reinsurance policies in a dynamic setting.
\newblock {\em Scandinavian Actuarial Journal}, 2001(1):55--68, 2001.

\bibitem{schmidli2002}
H.~Schmidli.
\newblock On minimizing the ruin probability by investment and reinsurance.
\newblock {\em Ann. Appl. Probab.}, 12(3):890--907, 08 2002.

\bibitem{schmidli:control}
H.~Schmidli.
\newblock {\em Stochastic Control in Insurance}.
\newblock Springer-Verlag, 2008.

\bibitem{schmidli:2018risk}
H.~Schmidli.
\newblock {\em Risk Theory}.
\newblock Springer Actuarial. Springer International Publishing, 2018.

\bibitem{wagner}
D.~Wagner.
\newblock Survey of measurable selection theorems.
\newblock {\em SIAM J. Control and Optimization}, pages 859--903, 1977.

\end{thebibliography}

\end{document}